\documentclass[manuscript,screen]{acmart}

\AtBeginDocument{%
  \providecommand\BibTeX{{%
    \normalfont B\kern-0.5em{\scshape i\kern-0.25em b}\kern-0.8em\TeX}}}

\setcopyright{acmcopyright}
\copyrightyear{XXXX}
\acmYear{XXXX}
\acmDOI{XXXXXXX.XXXXXXX}

\acmConference[TOCL]{ACM Transaction of Computational Logic}{}{}

\acmBooktitle{ACM Transaction of Computational Logic} 
\acmPrice{15.00}
\acmISBN{978-1-4503-XXXX-X/18/06}

\usepackage[utf8]{inputenc}
\usepackage[english]{babel}
\usepackage{amsthm}
\usepackage{amsmath}
\usepackage{amsfonts}
\usepackage{verbatim}
\usepackage{mathtools}
\usepackage{todonotes}
\usepackage{xspace}
\usepackage{epstopdf}
\usepackage{epstopdf}
\usepackage[ruled,linesnumbered]{algorithm2e}
\usepackage{algpseudocode}
\usepackage{graphicx}
\usepackage{tablefootnote}
\usepackage{lipsum}
\usepackage{mwe}

\newtheorem{thm}{Theorem}
\newtheorem{prop}[thm]{Proposition}

\newtheorem{lem}[thm]{Lemma}

\newtheorem{defn}{Definition}

\newtheorem{exm}{Example}

\newcommand{\system}{$S = (X, \iota(X), \tau(X, X'))\ $}
\newcommand{\tarray}{$A^E_I $}
\newcommand{\Lambdatool}{\textsc{Lambda}\xspace}
\newcommand{\cubicle}{\textsc{Cubicle}\xspace}
\newcommand{\zr}{\textsc{Z3}\xspace}
\newcommand{\mathsat}{\textsc{Mathsat}\xspace}
\newcommand{\mcmt}{\textsc{Mcmt}\xspace}
\newcommand{\mypyvy}{\textsc{MyPyvy}\xspace}
\newcommand{\icrpo}{\textsc{Ic3po}\xspace}
\newcommand{\Safe}{\textsc{Safe}\xspace}
\newcommand{\Unsafe}{\textsc{Unsafe}\xspace}
\newcommand{\defeq}{\vcentcolon=}
\newcommand{\UPDRIA}{\textsc{Updria}\xspace}
\newcommand{\icria}{\textsc{Ic3Ia}\xspace}

\def\checkmark{\tikz\fill[scale=0.4](0,.35) -- (.25,0) -- (1,.7) -- (.25,.15) -- cycle;} 

\tikzstyle{io} = [ellipse,thick,minimum width=1cm, minimum height=1cm, text centered, draw=black]
\tikzstyle{process} = [rectangle, minimum width=1cm, minimum height=1cm, text centered, draw=black, thick,
text width=3cm]
\tikzstyle{decision} = [diamond, thick,minimum width=1cm, minimum height=1cm, text centered, draw=black]
\tikzstyle{arrow} = [thick,->,>=stealth]
\tikzstyle{lemmas} = [cylinder, thick,shape border rotate=90, draw,minimum height=1cm,minimum width=1cm]

\usetikzlibrary{shapes.geometric, arrows, positioning, fit}

\begin{document}
\title{Invariant Checking for SMT-based Systems with Quantifiers}
\author{Gianluca Redondi}
\email{gredondi@fbk.eu}

\affiliation{%
  \institution{Fondazione Bruno Kessler \&}
  \institution{Universit\'a di Trento}
  \country{Italy}
}

\author{Alessandro Cimatti}
\email{cimatti@fbk.eu}
\affiliation{%
  \institution{Fondazione Bruno Kessler}
  \country{Italy}
}

\author{Alberto Griggio}
\email{griggio@fbk.eu}
\affiliation{%
  \institution{Fondazione Bruno Kessler}
  \country{Italy}
}

\author{Kenneth McMillan}
\email{kenmcm@cs.utexas.edu}
\affiliation{%
  \institution{University of Texas at Austin} 
  \country{USA}
}

\begin{abstract}
This paper addresses the problem of checking invariant 
properties for a large class of symbolic transition 
systems, defined by a combination of SMT theories and
quantifiers. State variables can be functions from an 
uninterpreted sort (finite, but unbounded) to
an interpreted sort, such as the the integers under the theory 
of linear arithmetic. This formalism
is very expressive and can be used 
for modeling parameterized systems, array-manipulating programs,
and more.  We propose two algorithms
for finding universal inductive invariants for such systems.  
The first algorithm combines an IC3-style loop with a form of 
implicit predicate abstraction to construct an invariant in an 
incremental manner. The second algorithm constructs an under-approximation
of the original problem, and searches for a formula which is an inductive 
invariant for this case; then, the invariant is generalized to the original case, 
and checked with a portfolio of techniques. We have implemented the two algorithms 
and conducted an extensive experimental evaluation, considering various 
benchmarks and different tools from the literature. As far as we know, our method 
is the first capable of handling in a large class of systems in a uniform way. 
The experiment shows that both algorithms are competitive with the state of the art.
\end{abstract}



\received{XXX}
\received[accepted]{XXX}

\maketitle


\section{Introduction}
Model checking algorithms are designed to automatically prove or disprove that a given property holds in a system. In particular, symbolic algorithms based on efficient quantifier-free SAT and SMT 
reasoning have seen significant progress in the last few years.
However, in many verification areas, first-order quantifiers are
needed, both in the symbolic description of the system and in the
property to prove. This is the case, for example, of verification of
parameterized systems, that consist of a finite but unbounded number of
components. Examples of such systems include communication protocols (e.g. leader election), feature systems, or control algorithms in various application domains (e.g. railways interlocking logic). The key challenge is to prove the correctness of the system for all
possible configurations corresponding to instantiations of the
parameters.

Unfortunately, dealing with the combined case of transition systems
with theories and quantifiers is far from trivial. Previous works on verification of parameterized systems with SMT-based techniques, such as~\cite{MCMT, cubicle}, impose strong syntactic restrictions on the formulae supported. Other approaches for model checking systems with quantifiers, such as~\cite{ic3po,updr,pdh},
are less restrictive, but they do not support theories.

In this paper, we
discuss the problem of model checking invariant properties of symbolic 
transition systems expressed in quantified first-order logic modulo theories, 
with state variables that are functions from finite finite but unbounded domains.

We start by defining our formalism, using the notion of \textit{array-based} transition systems 
(expanding the notion from \cite{SMTarray}), suited for the description of systems as above. 
Then, we propose two fully automated approaches for solving the invariant problem of array-based systems.

The first algorithm is based on UPDR \cite{updr}, an extension of the IC3/PDR algorithm that synthesize universally quantified invariants. Our procedure combines UPDR with implicit abstraction \cite{ImplicitAbstraction}, to deal with general SMT theories. 
Implicit abstraction was used in \cite{msatic3} to extend IC3 to deal with quantifier-free infinite-state systems. In this paper, we show how to modify that technique to handle quantified systems as well. We also lift several properties of the original algorithm to our setting. However, the resulting algorithm heavily relies on quantified queries to an external solver: in practice, this is still an expensive subroutine. 

The second algorithm consists of a simple yet general procedure based on the interaction of two key ingredients.
First, we restrict the cardinality of the uninterpreted sort to a fixed natural number. 
This results in a quantifier-free system that can be model-checked with existing techniques: upon termination, we either get a counterexample, in
which case the system is unsafe, or a proof for the property in the form of a quantifier-free inductive invariant. 
Second, the invariant is generalized to a quantified candidate invariant for the original system, and its validity is checked 
using either a form of abstraction, or quantified SMT reasoning. If the candidate invariant is valid, then the system is safe. Otherwise, further reasoning is required, e.g. by increasing the cardinality of the domain, and iterating the first step. 
The first option to check that a candidate invariant holds is based on an extension of the works on Parameter Abstraction of~\cite{SimpleMethod,Krstic2005ParametrizedSV} to our formalism. This technique computes a quantifier-free system and a quantifier-free property that, in case it holds, ensures that the original property is an invariant of the system. 
As a second option, we could use instead any off-the-shelf solver supporting SMT and quantifiers (e.g. Z3\cite{z3}) to discharge validity checks. However, a black-box approach to checking the validity of quantified formulae may cause the procedure to diverge in practice. Therefore, we also propose a more careful, resource-bounded approach to instantiation, that can be used to discharge quantified queries in a more controlled way.

Both algorithms have been implemented and experimentally evaluated over a large benchmark set, obtained from different sources, and making use of theories, quantifier alternations, or both. The experimental evaluation demonstrates that the algorithms are very general and that they compare well with other tools in each category.  As far as we know, our algorithms are the first capable of approaching in a
uniform way such a large variety of systems.

We remark that the second algorithm of this paper is a combination of the algorithms given in 
\cite{Lambda, atva22}, described in a unified manner and experimentally evaluated over a larger set of benchmarks.

This paper is structured as follows. In Section \ref{sec-preliminaries}, we present some background. In Section \ref{sec-problem}, we introduce our formalism, and we state the problem faced in this work. In Section \ref{sec-updria}, we illustrate the first algorithm, while in Section \ref{sec-lambda}, we describe the second one. We discuss related work in Section \ref{sec-rel-work}, and we conduct the experimental evaluation of both algorithms, as well as a comparison with other tools, in Section \ref{sec-expval}. Finally, we draw conclusions in Section \ref{sec-conclusions}.
\section{Background}
\label{sec-preliminaries}
\subsection{Preliminaries}
A theory $\mathcal{T}$ in the SMT sense is a pair 
$\mathcal{T} = (\Sigma, \mathcal{C})$, 
where $\Sigma$ is a first-order signature and $\mathcal{C}$ is a class of models over $\Sigma$. 
We use the standard notions of first-order logic and Tarskian interpretation  (variables, assignment, model, satisfiability, validity, logical consequence). In this paper, we will work in a multi-sorted setting: we outline here the main differences with the standard case. A sorted signature $\Sigma$ is given by a set of sorts and a set of sorted symbols, such that the domain and range sorts of every symbol in $\Sigma$ is also in $\Sigma$. A multi-sorted structure is a map that takes every sort symbol to a set called its universe (disjoint from other universes) and every function symbol to a function between the universes of the corresponding sorts.

Given a model $\mathcal{M}$ over $\Sigma$, the restriction of $\mathcal{M}$ over a sorted vocabulary $\Sigma' \subset \Sigma$ is given by the domain restriction of $\mathcal{M}$ to $\Sigma'$.  Given two theories $\mathcal{T} = (\Sigma, \mathcal{C})$ and $\mathcal{T'} = (\Sigma', \mathcal{C}')$, the \textit{combination} of the two theories is a theory whose signature is $\Sigma \cup \Sigma'$, and a model for it is a structure such that its restriction onto $\Sigma$ is in $\mathcal{C}$ and its restriction onto $\Sigma'$ is in $\mathcal{C'}$.

A literal is an atom or its negation. A clause is a disjunction of literals. A ground term is a term that does not 
contain variables. A formula is in conjunctive normal form (CNF) iff it is a conjunction of clauses. 
If $\Sigma$ is a signature, and $X$ are a set of symbols not contained in $\Sigma$, we write $\phi(X)$ to denote a formula in the new signature expanded with such symbols (denoted with $\Sigma(X)$).

If $\phi$ is a formula, $t$ is a term and $x$ is a symbol that occurs in $\phi$, we write $\phi[x/t]$ for the substitution of every occurrence of $x$ with $t$. If $T$ and $X$ are vectors of the same length, we write $\phi[X/T]$ for the
simultaneous substitution of each $x_i$ with the corresponding term $t_i$. If $F$ denotes a set of formulae, with a slight
abuse of notation, we may use again the symbol $F$ to denote instead the formula given by the conjunction of all
the elements in $F$.

In Boolean logic, models can be represented by Boolean formulae simply by considering a conjunction between literals 
that are true in the model. A similar concept in first-order logic is the one of diagram. Diagrams are a method 
to represent a set of finite first-order models with an existentially quantified formula. We report the notion from \cite{updr}:
\begin{defn}
\label{defn-diagram}
Given a finite model $\mathcal{M}$, let $x_{m_i}$ be a fresh variable for each 
element $m_i$ in the universe of the model. The diagram of $\mathcal{M}$ over $\Sigma$ is the existential closure of the conjunction of the following literals: 
\begin{itemize}
\item[(i).] inequalities of the form $x_{m_i} \neq x_{m_j}$ for every pair of distinct elements $m_i, m_j$ in the model;
\item[(ii).] equalities of the form $c = x_m$ for every constant symbol $c$ in $\Sigma$ such that $\mathcal{M} \models c = m$;
\item[(iii).] the atomic formula $P(x_{m_{i_1}}, \dots, x_{m_{i_n}})$ for every predicate symbol $P$ of ariety $n$, if $\mathcal{M} \models P(m_{i_1}, \dots, m_{i_n})$, or the atomic formula $\neg P(m_{i_1}, \dots, m_{i_n})$ otherwise;
\item[(iv).] the atomic formula $f(x_{m_{i_1}}, \dots, x_{m_{i_n}}) = x_m$ for every function symbol $f$ of ariety $n$ if $\mathcal{M} \models f(m_{i_1}, \dots, m_{i_n}) = m$.
\end{itemize} 
\end{defn}

\subsection{Symbolic transition systems}
In the following, we suppose fixed a theory $\mathcal{T} = (\Sigma, \mathcal{C})$.
Given a set of constants $X$, we denote with $X'$ the set $\{x'\ |x\in X\}$. If $F(X)$ is a formula, we denote with $F'$ the formula $F[X/X']$. A symbolic transition system is a triple $( X,\iota(X), \tau(X, X'))$, where $X$ is a set of constants, called the \textit{state variables} of the system,  
and $\iota(X)$, $\tau(X, X')$ are $\Sigma(X, X')$-formulae.  
A state is a given by a model $\mathcal{M}$ for $\mathcal{T}$, and a valuation $s$ of the state variables $X$ in the universe of $\mathcal{M}$.
 A state is initial iff it is a model of $\iota(X)$, i.e. $\mathcal{M}, s\models \iota(X)$. A couple of states $(\mathcal{M}, s)$, $(\mathcal{M}, s')$ denote a transition iff
$\mathcal{M}, s, s'\models \tau(X, X')$, also denoted as $T(s, s')$. When clear from the context, we may omit the model $\mathcal{M}$ from the notation, and consider it fixed.
A variable $x$ is a \textit{frozen} variable
iff its value is fixed during every evolution of the system, i.e. if the constraint $x' = x$ is implied by $\tau$.
A path is a sequence of states $s_0, s_1, \ldots$ such that $s_0$ is initial and $T(s_i,s'_{i+1})$ for all $i$. If $\pi$ is a path, we denote with $\pi[j]$ the $j$-th element of $\pi$. A state $s$ is reachable iff there exists a path $\pi$ such that $\pi[i]=s$ for some $i$.

A formula $\phi(X)$ is an invariant of the transition system $C = ( X,\iota(X),$ $\tau(X, X'))$ iff it holds in all the reachable states. 
Following the standard model checking notation, we denote this with $C\models\phi(X)$.%
\footnote{Note  that  we  use  the  symbol $\models$  with  three  different  denotations: if $\phi,\psi$ are formulae, $\phi \models \psi$ denotes that $\psi$ is a logical consequence of $\phi$; if $\mu$ is an interpretation, and $\psi$ is a formula, $\mu \models \psi$ denotes that $\mu$ is a model of $\psi$; if $C$ is a transition system, $C \models \psi$ denotes that $\psi$ is an invariant of $C$.}
A formula $\phi(X)$ is an inductive invariant for $C$ iff  $\iota(X) \models \phi(X)$ and 
$\phi(X)\wedge \tau(X, X')\models \phi(X')$. 
Given a formula $\psi(X)$, we say that a formula $\phi$ is an inductive invariant for $\psi$ and $C$ if $\phi$ is an 
inductive invariant and $\phi \models \psi$.

\subsection{Approximate Reachability Sequence}
We introduce here some standard concept to simplify the exposition of one of the algorithms, which is based on IC3\cite{ic3}. 
The algorithm maintains a \textit{trace}, i.e. an ordered sequence $F_0, \dots, F_N$ of formulae, called frames, where the formula $F_i$ over-approximate the set of states reachable in up to $i$ transitions. More formally, the trace of IC3 is an \textit{approximate reachability sequence}:
\begin{defn}
\label{def-approx-reach}
    Let $C = ( X,\iota(X), \tau(X, X'))$ be a symbolic transition system and $\phi(X)$ a formula. A sequence of formulas $F_0(X), \dots, F_N(X)$ is an \textit{approximate reachability sequence} for $C$ and $\phi$ if:
    \begin{itemize}
        \item $\iota(X) \models F_0(X)$;
        \item $F_i(X) \models F_{i+1}(X)$ for all $0 \leq i < N$;
        \item $F_i(X) \wedge \tau(X, X') \models F_i(X')$ for all $0 \leq i < N$;
        \item $F_i(X) \models \phi(X)$ for all $0 \leq i < N$.
    \end{itemize}
\end{defn}
An approximate reachability sequence can be used to prove that a formula $\phi$ is an invariant for a system, thanks to the following:
\begin{prop}[~\cite{ic3}]
\label{prop-approxi-sequence}
  Let $C = ( X,\iota(X), \tau(X, X'))$ be a transition system and $\phi$ a formula. Let $F_0, \dots, F_N$ be an \textit{approximate reachability sequence} for $C$ and $\phi$. If for some $0 \leq i < N, F_{i+1} \models F_i$, then $C \models \phi$. Moreover, $F_i$ is an inductive invariant for $\phi$ and $C$.
\end{prop}

\subsection{Abstraction and Refinement}
\label{subs-simul}
Abstraction is one of the most common tools in model checking. Often, one can consider a simplified version of a problem by removing some information, and reducing it to a simpler system. More formally, the notion of abstraction is embodied by the definition of simulation:
\begin{defn}[Simulation]
\label{def-simul}
  Given two symbolic transition systems $C_1 = (X_1,\iota_1,\tau_1 )$ and
  $C_2 = (X_2,\iota_2,\tau_2)$,
  with sets of states $S_1$ and $S_2$, a \textit{simulation} $\mathcal{S}$
  is a relation $\mathcal{S} \subset S_1 \times S_2$, such that:
\begin{itemize}
\item for every $s_1 \in S_1$ such that $s_1 \models \iota_1$,
 there exists some $s_2 \in S_2$ such that $(s_1, s_2)\in\mathcal{S}$ and $s_2\models\iota_2$;
\item for every $(s_1,s_2)\in\mathcal{S}$, and for every $s_1' \in
  S_1$ such that $s_1, s_1'\models \tau_1$, there some
  $s_2' \in S_2$ such that
  $(s'_1,s_2') \in \mathcal{S}$ and  $s_2, s_2'\models\tau_2$.
\end{itemize} 
\end{defn}

It is easy to see that simulation preserves reachability: if there exists a state $s$ of $C_1$ reachable by a path $\pi$ of length $k$, then there exists a state $s'$ of $C_2$ reachable by a path $\pi'$ of length $k$ such that, for every $i$, $(\pi[i], \pi'[i]) \in \mathcal{S}$. Often, simulation preserves the truth values of certain atomic propositions; in these cases, simulations preserve also the validity of reachability properties that are defined over such atoms.
A weaker notion of simulation that still preserves reachability (even if with paths of different lengths) is the following:

\begin{defn}[Stuttering Simulation]
\label{def-stutter-simul}
  Given two symbolic transition systems $C_1 = (X_1,\iota_1,\tau_1 )$ and
  $C_2 = (X_2,\iota_2,\tau_2)$,
  with sets of states $S_1$ and $S_2$, a \textit{stuttering simulation} $\mathcal{S}$
  is a relation $\mathcal{S} \subset S_1 \times S_2$, such that:
\begin{itemize}
\item for every $s_1 \in S_1$ such that $s_1 \models \iota_1$,
 there exists some $s_2 \in S_2$ such that $(s_1, s_2)\in\mathcal{S}$ and $s_2\models\iota_2$;
\item for every $(s_1,s_2)\in\mathcal{S}$, and for every $s_1' \in
  S_1$ such that $s_1, s_1'\models \tau_1$, there exists either some
  $s_2' \in S_2$ such that
  $(s'_1,s_2') \in \mathcal{S}$ and  $s_2, s_2'\models\tau_2$, or some $(s_2', s_2'') \in S_2 \times S_2$ such that $(s'_1,s_2'') \in \mathcal{S}$, and  $s_2, s_2'\models\tau_2$,  $s_2', s_2''\models\tau_2$.
\end{itemize} 
\end{defn}

One of the most used abstraction paradigms in model checking of infinite-state systems is Predicate Abstraction. This technique is usually used to reduce to the Boolean case: given an infinite-state system $C$ and a set of predicates, one can build a new system $\hat{C}$, purely propositional, such that $\hat{C}$ simulates $C$. Intuitively, the new system $\hat{C}$ is obtained with a quotient of $C$, by considering only the truth values of the fixed set of predicates.  
In general, the explicit computation of the system $\hat{C}$ can be very expensive, requiring typically ALLSAT techniques \cite{ALLSAT}. However, more efficient encodings are possible if the simulation relation can be expressed in logic \cite{ImplicitAbstraction, DBLP:conf/cav/LahiriBC05}.

\subsubsection{Refinement by Craig interpolants}
\label{subs-interpol}
Usually, abstractions are not enough to prove the desired property, and spurious counterexamples may occur. In this case, the abstraction needs to be refined to eliminate the abstract counterexample. Therefore, most abstraction-based algorithms perform CEGAR (counterexample abstraction refinement)\cite{CEGAR}: abstract counterexamples are continuously analyzed to improve the precision of the abstraction. Since the work of McMillan \cite{Lazyabs}, a common approach for CEGAR loops for symbolic transition systems is based on interpolants. 
\begin{defn}
Given an ordered set of formulae $(\phi, \psi)$ in a theory $\mathcal{T}$, a Craig interpolant is a formula $\iota$ such that:
\begin{itemize}
    \item[(i)] $\phi \models \iota$;
    \item[(ii)] $\psi \wedge \iota$ is -unsatisfiable;
    \item[(iii)] all the uninterpreted symbols occurring in $\iota$ occur in both $\phi$ and $\psi$.
\end{itemize}
\end{defn}
This definition can be extended to an ordered sequence of formulae $\phi_0, \dots, \phi_n$ 
such that their conjunction is unsatisfiable, obtaining a \textit{sequence interpolant} $\iota_1, \dots, \iota_n$ 
such that, for all $i$:
\begin{itemize}
    \item[(i)] $\bigwedge_{0 \leq k < i} \phi_k \models \iota_i$;
    \item[(ii)] $\bigwedge_{i \leq k \leq n }\phi_k \wedge \iota_i$ is unsatisfiable;
    \item[(iii)] $\iota_i \wedge \phi_{i+1} \models \iota_{i+1}$ for all $1 \leq i < n$;
    \item[(iv)] all the uninterpreted symbols occurring in $\iota_i$ occur in both $\bigwedge_{0 \leq k < i} \phi_k$ and $\bigwedge_{i \leq k \leq n }\phi_k$.
\end{itemize}
Interpolants and sequence interpolants can be efficiently computed in several important theories with several techniques. e.g. see \cite{InterpolMathsat}.

In the context of predicate abstraction, given an abstract counterexample $\pi$, the predicates given by a sequence of interpolants are guaranteed to rule out the counterexample, once they are added to the set of predicates for the abstraction \cite{Lazyabs}.


\section{Problem statement}

\label{sec-problem}
The problem discussed in this paper is the problem of invariant checking over a class of symbolic transition systems. Such a  class is a generalization of the formalism introduced in \cite{MCMT, Ivy2016}, especially suited for the verification of parameterized systems. We call this class \textit{array-based transition systems}, following \cite{MCMT}, even if our formalism is more general, allowing general transition formulae. Note that this is also different from standard symbolic transition systems over the extensional theory of arrays.

\subsection{Array-based transition systems}

We start by considering two theories; a theory $\mathcal{T_I} = (\Sigma_I, \mathcal{C}_I)$, called the $index$ theory, whose class of models does not contain infinite universes. In practice, this is often the theory of an uninterpreted sort with equality, whose class of models includes all possible finite (but unbounded) structures. 
In addition, we consider a quantifier-free theory of \textit{elements} $\mathcal{T}_E = (\Sigma_E, \mathcal{C}_E)$, used to model the data of the system. Relevant examples consider as $\mathcal{T}_E$ the theory of an enumerated datatype, linear arithmetic (integer or real), or a combination of those. In addition, we consider an array theory, whose signature only contains the function symbol $\{\cdot [\cdot]\}$, which is interpreted as the function application. Finally,  with \tarray{} we denote the combination of the index, the element, and the array theory. Therefore, the signature of \tarray{} is $\Sigma = \Sigma_I \cup \Sigma_E \cup \{\cdot [\cdot]\}$, 
and a model for it is given by a set of total functions from the universe of a model of the index sort, to a universe of the element sort, and $\cdot[\cdot]$ is interpreted as the function application.  In the following, we will denote with letters $I, J$ variables of index sort, while we use the letter $X$ to denote array symbols.
We restrict ourselves to one index theory and one element theory for the sake of simplicity, but typically applications include a combination of several index theories and several element theories. 

\begin{defn}
\label{TS}
An array-based transition systems is a triple $$S = (X, \iota(X), \tau(X, X'))$$
where:
\begin{itemize}
    \item[(i.)] $X$ is a set of symbols of array sorts;
    \item[(ii.)] $\iota(X), \tau(X, X')$ 
      are $\Sigma(X, X')$-formulae  (possibly containing quantified  variables only of sort $\mathcal{T}_I$).
\end{itemize}
\end{defn}

Note that arrays can be also 0-ary or constant functions, thus allowing to have state variables of index or element sort.

\begin{exm}{}
We show how to model a simplified version of the Bakery algorithm in our formalism. We use as $\mathcal{T}_I$ the theory of equality. As for $T_E$, we need a combination of two theories: an enumerated datatype with values $\{idle, wait, crit\}$, and $\mathcal{LIA}$ (linear integer arithmetic). 
We define four state variables: one array $state$ with values in  $\{idle, wait, crit\}$, one array $ticket$ with values in $\mathbb{Z}$, and two integer variables $next\_ticket$ and $to\_serve$. The initial formula of the system is: \[ \forall i.state[i] = idle \wedge ticket[i] = 0 \wedge next\_ticket = 1 \wedge to\_serve = 1 \ \]
The transition formula is the disjunction of the three formulae; the first one models a process that goes from idle to wait and takes a ticket:
\begin{align*}
    \exists i. \big(state[i] = idle \wedge state'[i] = wait \wedge t'[i] = next\_ticket \wedge {}\\
     next\_ticket' = next\_ticket + 1 \wedge to\_serve' = to\_serve \\ {} \wedge \forall j. (j \neq i \rightarrow state'[j] = state[j] \wedge ticket'[j] = ticket[j])\big).
\end{align*}
In the second disjunct, a process enters the critical section if its ticket is selected:
\begin{align*}
    \exists i. \big(state[i] = wait \wedge state'[i] = crit \wedge t[i] = to\_serve \\ 
    {} \wedge next\_ticket' = next\_ticket \wedge to\_serve' = to\_serve \wedge \\
    {} \forall j. (j \neq i \rightarrow state'[i] = state[i])  \wedge (\forall j. ticket'[j] = ticket[j])\big).
\end{align*}
Finally, a process exits from the critical state and reset its ticket:
\begin{align*}
    \exists i. \big(state[i] = crit \wedge state'[i] = idle \wedge t'[i] = 0 \wedge\\
    {} next\_ticket' = next\_ticket \wedge to\_serve' = to\_serve + 1 \wedge \\
    {}\forall j. (j \neq i \rightarrow state'[j] = state[j] \wedge ticket'[j] = ticket[j])\big).
\end{align*}
\qed
\end{exm}

Given a model $\mathcal{M}$ of \tarray, a state of an array-based transition system is a valuation of the elements in $X$ in $\mathcal{M}$, i.e. an assignment of the state variables to functions from a finite universe of the index sort, to an universe of an element sort.  

The presence of quantifiers in the theory $\mathcal{T}_I$ make the general problem of satisfiability of \tarray{} formulae to be undecidable (even in the case where the element sort is boolean). However, a common fragment of first-order logic used to describe parameterized problems is EPR \cite{FMCADepr} (Effective Propositional Logic), which is decidable. In our setting, we have the following decidability result, which combines classical EPR with ground SMT theories:
\begin{prop}
\cite{MCMT}
\label{prop-decidable-epr}
    If $\Sigma_I$ is relational, and the quantifier-free fragment of the theories $\mathcal{T}_I$ and $\mathcal{T}_E$ are decidable, then the \tarray-satisfiability of formulae of the form
    \begin{equation}
    \label{epr}
      \exists I. \forall J. \phi(I, J, X)  
    \end{equation}
    is decidable.
\end{prop}


\subsubsection{The invariant problem for array-based transition systems}
\label{subs-subs-problem}
Let $S$ be an array-based transition system, and $\phi(X)$ a $\Sigma(X)$-formula, possibly containing quantified variables of sort $\mathcal{T}_I$. The \textit{Invariant Problem} we consider is the problem of deciding whether $S \models \phi(X)$.  The problem is in general undecidable since it subsumes undecidable problems such as safety of parameterized systems \cite{MCMT, ParamDecidab, DecidabInvariants}.
\begin{exm}
    Following the last example, the property we want to prove is mutual exclusion: 
\[ \forall i,j. (i \neq j \rightarrow \neg (state[i] = crit \wedge state[j] = crit)). \]
\end{exm}

\subsection{Ground instances}
\label{subs-ground} 
If a finite cardinality $n$ for the $\mathcal{T}_I$-models is fixed, one can construct a quantifier-free under-approximation of $S$ by considering another symbolic transition system, called ground $n$-instance, by considering as states functions with, as domains, sets of only that size. 
In this section, we define how to construct such a system.
First, we consider $C = \{c_1, \dots, c_n\}$ a set of fresh constants of index sort. These will be frozen variables 
of the ground instance; moreover, they will be also considered all implicitly different. In the following, if $\phi(X)$ is a formula with quantifiers of only sort $\mathcal{T}_I$, we denote $\phi_n(C, X)$ the quantifier-free formula obtained by grounding the quantifiers in $C$, i.e. by recursively applying the rewriting rules:
\begin{align}
    \label{rewriting-rules}
    \forall i. \phi'(i, X) \mapsto \bigwedge_{k=1}^n \phi'(i, X)[i/c_k] \\
    \exists i. \phi'(i, X[i]) \mapsto \bigvee_{k=1}^n \phi'(i, X)[i/c_k]
\end{align}    

Finally, we also need to restrict function symbols with values in $\mathcal{T}_I$. To do so, for each function symbol $a$ in $\Sigma_I$ whose codomain type is $\mathcal{T}_I$, we define the formula
\[\alpha(a) \defeq \forall i_1, \dots, i_m \exists j. a(i_1, \dots, i_m) = j,\]
where $m$ is the ariety of $a$, and $i_1, \dots, i_m, j$ are fresh variables of index sort.

\begin{defn}[Ground $n-$instance]
Let $S = (X, \iota(X), \tau(X, X'))$ be an array-based transition system, and $n$ an integer. Let $\bar{\iota}(X) \defeq \iota(X) \wedge \bigwedge_{a} \alpha(a)_n$, and $\bar{\tau}(X, X') \defeq \tau(X, X') \wedge \bigwedge_{a} \alpha(a)_n$, where $a$ ranges over each function symbol in $\Sigma_I$ whose codomain type is $\mathcal{T}_I$.
The ground $n$-instance of the system $S$ is a symbolic transition system $S_{n}$ defined by:
\[S_{n} = \big(C \cup X, \bar{\iota}_n(C, X), \bar{\tau}_n(C, X, X')\wedge C' = C  \big). \]
\end{defn}
To simplify the use of notation, we will denote with $\iota_n$ and $\tau_n$ the initial and transition formula of the $n$-instance $S_{n}$. Observe that a state of $S_{n}$ is given by: $(i)$ a valuation of the symbols $C$ in a finite index universe, and $(ii)$ an interpretation of the state variables X as functions from that universe to an element sort. Note that even if the models of $T_I$ are all finite, the set of states of $S_{n}$ can be infinite, since $\mathcal{T}_E$ could have an infinite model, e.g. if integer or real variables are in the system. 

We note that, due to restriction on finite index models, it follows from the definition that the invariant problem for an array-based transition system, $S \models \phi$, is equivalent to checking if $S_{n} \models \phi_n$ for all possible $n$.

\section{UPDR with implicit abstraction}
\label{sec-updria}
\subsection{Overview}
The paper \cite{updr} proposes a variant of the PDR/IC3 \cite{ic3} algorithm to solve the problem stated in \ref{subs-subs-problem}, but only in the case of $\mathcal{T}_E$ equal to the theory of Boolean. The algorithm, called UPDR (Universal Property Directed Reachability), represents first-order models by formulae with the notion of diagrams (see preliminaries). 
When extending UPDR to our general setting, a main challenge is how to extend the computation of diagrams, to handle a general SMT theory $\mathcal{T}_E$. In this section, we combine UPDR with a form of predicate abstraction, so that every diagram computed by the procedure will still be over models of index sort only. 

First, we introduce the necessary notions for the predicate abstraction, and how to use them in our context. Then, we illustrate the general algorithm, and we discuss some approaches that we used for its implementation. Finally, we study its theoretical properties, most of which are a direct lifting of the ones given in \cite{updr}. The proofs of the main results are given in the appendix. 

\subsection{Implicit Indexed Predicate Abstraction}
\label{subs-index-predicate-abs}

Predicate abstraction has been commonly considered for quantifier-free systems, but indexed predicates were introduced to abstract some restricted classes of systems with quantifiers \cite{indexpredicates}.
As already mentioned in the preliminaries, building the abstract version of a system can be expensive; however, various techniques that encode the abstract transition relation with logical formulas, such as \cite{ImplicitAbstraction, DBLP:conf/cav/LahiriBC05}, have proven to be successful in the quantifier-free and Boolean case. 
In this section, we introduce a form of abstraction suitable for our purposes: instead of Boolean variables, we consider a set of first-order predicates, where free variables may occur. 

In the following, we fix a set $I=\{i_1, \cdots, i_k\}$ of variables of index sort, and a finite set $\mathcal{P}(I) =$ $\{p_1(I, X), \dots, p_n(I, X) \}$  of $\Sigma(I)$ atoms (i.e. predicates over $\Sigma$ possibly containing free variables $I$ of sort $\mathcal{T}_I$). 
 We can write $\mathcal{P}$ instead of $\mathcal{P}(I)$ for the sake of simplicity.\\ 
\begin{exm}
In the bakery example, we can consider as a set of index predicates: 
\[\mathcal{P}(i_1) = \{ state[i_1] = idle, ticket[i_1] = 0, next\_ticket = 1, to\_serve = 1, state[i_1] = crit \}, \]
where $i_1$ is a free variable of index sort.

\end{exm}

We now introduce and generalize the main concepts of implicit abstraction for the quantified case.
\begin{defn}
Given $\mathcal{P}(I)$ a set of index predicates, we define:
\begin{itemize}
    \item $X_{\mathcal{P}(I)}$ is a set of fresh $\mathcal{T}_I$ predicates, one for each element $ p \in \mathcal{P}(I)$, and with ariety the number of free variables in $p$:
    \[  X_{\mathcal{P}(I)} \defeq \{ x_{p(I, X)}(I)\mid p(I, X) \in \mathcal{P}(I) \};\]
    \item the formula \[ H_\mathcal{P}(X_{\mathcal{P}}, X) \defeq \forall I. \big(\bigwedge_{p(I, X) \in \mathcal{P}} x_{p(I, X)}(I) \leftrightarrow p(I, X)\big)  \]
\item the formula
\[ EQ_\mathcal{P}(X, X') \defeq \forall I. \big(\bigwedge_{p(I, X) \in \mathcal{P}} p(I, X) \leftrightarrow p(I, X'))  \]
\end{itemize}
\end{defn}
In the above definition, the new predicates $X_\mathcal{P}$ will act as new state variables of the abstract system. The formula $H_\mathcal{P}$ is used to define the simulation relation between an abstract and a concrete state. Finally, the formula $EQ_\mathcal{P}$ is used to compute the abstraction implicitly.
\begin{exm}
Considering again the bakery algorithm, given the set of predicates of the previous example, we have a set of abstract variables
\begin{align*}
   X_{\mathcal{P}(i_1)} = \{ x_{state[i_1] = idle}(i_1), x_{ticket[i_1] = 0}(i_1), 
   x_{next\_ticket = 1}, 
   x_{to\_serve = 1},  \\
x_{state[i_1] = crit}(i_1) \}
\end{align*}
Moreover, we have that $H_\mathcal{P}(X_{\mathcal{P}}, X)$ is equal to:\begin{align*}
     \forall i_1. (x_{state[i_1]= idle}(i_1) \leftrightarrow state[i_1] = idle) \wedge \forall i_1. (x_{ticket[i_1] = 0}(i_1) \leftrightarrow ticket[i_1] = 0) \\ {} \wedge x_{next\_ticket = 1} \leftrightarrow next\_ticket = 1 \wedge x_{to\_serve = 1} \leftrightarrow to\_serve = 1 \wedge \\ {} \forall i_1. (x_{state[i_1] = critical}(i_1) \leftrightarrow state[i_1] = critical)
\end{align*}
\end{exm}

Given a formula $\phi$, the predicate abstraction of $\phi$ with respect to $\mathcal{P}$, denoted $\hat{\phi}_{\mathcal{P}}$, is obtained by adding the abstraction relation to it and then existentially quantifying the variables $X$, i.e.,  $\hat{\phi}_{\mathcal{P}}(X_{\mathcal{P}}) \defeq \exists X. (\phi(X) \wedge  H_\mathcal{P}
(X_{\mathcal{P}}, X) )$, and similarly for a (transition) formula over $X$ and $X'$ we have $\hat{\phi}_{\mathcal{P}}(X_{\mathcal{P}}, X'_{\mathcal{P}}) \defeq \exists X, X'.  (\phi(X, X') \wedge  H_\mathcal{P} (X_{\mathcal{P}}, X) \wedge H_\mathcal{P} (X'_{\mathcal{P}}, X') )$. Moreover, given $\phi$, we denote with $\bar{\phi}$ the formula obtained by replacing the atoms in $\mathcal{P}$ with the corresponding predicates. Dually, if $\phi$ is a formula containing the $X_{\mathcal{P}}$ predicates, with $\phi[X_{\mathcal{P}} / \mathcal{P}]$ we denote the formula obtained by restoring the original predicates in place of the abstract ones. Remark that these are not substitutions, since the index variables may have been renamed, and thus beta-conversion could be needed. 
It is easy to see that, for any formula $\psi$, if all atoms of $\psi$ occur in $\mathcal{P}$, then $\hat{\psi}$ and $\bar{\psi}$ are logically equivalent under the assumption $\mathcal{H}_{\mathcal{P}}$.
The indexed predicate abstraction of a system \system is obtained by abstracting the initial and the transition conditions, i.e. 
$\hat{S}_\mathcal{P} = (X_\mathcal{P}, \hat{\iota}_\mathcal{P}(X_\mathcal{P}), \hat{\tau}_\mathcal{P}(X_\mathcal{P}, X'_\mathcal{P}))$. When clear from context, we will omit the subscript $\mathcal{P}$. 
We have:
\begin{prop}
\label{prop-index-abs-1}
   Let \system an array-based transition system, and $\phi(X)$ a formula. Let $\mathcal{P}$ a set of index predicates which contains all the atoms occurring in $\phi$. If $\hat{S}_\mathcal{P} \models \hat{\phi}$ then $S \models \phi$.
\end{prop}

In our algorithm, we start with a system $S$ and use a combination of UPDR, implicit abstraction, and indexed predicates to check if the property holds in an abstract system $\hat{S}$. 
However, the abstraction may be too coarse, and spurious counterexamples can occur. Notably, such counterexamples are actually a proof that no universal invariant exists over the set of predicates $\mathcal{P}$: this result is a direct consequence of the fact that our algorithm is a lifting of the original UPDR\cite{updr}.

In the following, when working in the abstract space of $S$, the critical step for the algorithm is repeatedly checking whether a formula representing a model $\psi$ is inductive relative to the frame $F$. Frames will be defined as a set of negation of diagrams (except for the initial frame $F_0$ which is $\hat{\iota}$). The insight underlying implicit abstraction is to perform the check
without actually computing the abstract version of $\tau$, but by encoding the simulation relation with a logical formula. This is done by checking
the formula:
\begin{align*}
     AbsRelInd(F, \tau, \psi, \mathcal{P}) =  F(X_{\mathcal{P}}) \wedge \psi(X_{\mathcal{P}})\wedge H_\mathcal{P}(X_{\mathcal{P}}, X)
      \\
      {} \wedge EQ_\mathcal{P}(X, \bar{X}) \wedge \tau(\bar{X}, \bar{X}')\wedge EQ_\mathcal{P}(\bar{X}', X') \wedge \neg \psi(X_{\mathcal{P}}') \wedge H_\mathcal{P}(X_{\mathcal{P}}', X')  
    \end{align*}

We have: 
\begin{prop}
\label{prop-rel-inductive-1}
Let \system and let $\hat{S}_\mathcal{P}$ be its indexed predicate abstraction. Given any formulae $F, \psi$, then the formulae
    $AbsRelInd(F, \tau, \psi, \mathcal{P})$ and $F(X_{\mathcal{P}}) \wedge \hat{\tau}(X_{\mathcal{P}}, X'_{\mathcal{P}}) \wedge \psi(X_{\mathcal{P}}) \wedge \neg \psi(X'_{\mathcal{P}})$ are equisatisfiable. 
\end{prop}
    
\subsection{Algorithm description and pseudocode}
\label{subs-updria-alg}
We can now describe the whole procedure, depicted in the Algorithm \ref{updralg}.
\begin{algorithm}[ht]
\SetAlgoLined
\caption{UPDR + IA}
\label{updralg}
Input: $S = (X, \iota(X), \tau(X, X')), \phi(X)$
\\
    $\mathcal{P}(I) = \{ \textit{set of atoms occuring in } \iota, \phi\}$
\\
\textbf{if not} $\bar{\iota} \wedge H_P \models \bar{\phi}$:\\
\Indp \textbf{return} cex \textit{ \# cex in initial state}
\\
\Indm
$F_0 = \{\bar{\iota}\}$ 
\\
$k = 1, F_k = \emptyset $ \\
\textbf{while} True: \\
\Indp \textbf{while} $F_k \wedge H_P \wedge \neg \bar{\phi}$ is sat: \\
\textit{\# let $\sigma$ be an \tarray{} model, and $\sigma_I$ its restriction on the index sort;} \\
\Indp $\psi = diag(\sigma_I)$\\
\textbf{if not} $RecBlock(\psi, k)$: \\
\Indp 
\textit{\# an abstract counterexample $\pi$ is provided} \\
\textbf{if not} $Concretize(\pi)$: \\
\Indp $\mathcal{P} = \mathcal{P} \cup Refine(\pi)$\\
\Indm \textbf{else:}\\
\Indp \textbf{return} \textbf{unsafe} \\
\Indm \Indm \Indm
\textit{\# if no models, add new frame:}\\
$k = k+1, F_k = \emptyset $ \\
\textit{\# propagation phase}\\
\textbf{for} $i=1, \dots, k-1$:\\
\Indp \textbf{for each diagram $\psi \in F$}:\\
\Indp \textbf{if} $AbsRelInd(F_i, \tau, \psi, \mathcal{P})$ is unsat: \\
\Indp add $\psi$ to $F_{i+1}$\\
\Indm
\textbf{if} $F_i = F_{i+1}$: \\
\Indp \textbf{Return} \textbf{safe}
\end{algorithm}

\begin{algorithm}[ht]

\caption{$RecBlock(\psi, N)$}
\label{recblock}
\textbf{if $N = 0$} \\
\Indp \textbf{return} False\\
\Indm \textbf{While} $AbsRelInd(F_{i-1}, \tau, \neg \psi, \mathcal{P})$ is sat:\\
\textit{\# let $\sigma'$ be an \tarray{} model, and $\sigma'_I$ its restriction on the index sort;} \\
\Indp $\psi' = diag(\sigma'_I)$\\
\textbf{if not} $RecBlock(\psi', N-1)$:\\
\Indp \textbf{return} cex\\
\Indm \Indm
$g = Generalize(\neg \psi, N)$ \\
add $g$ to $F_1, \dots, F_N$ \\
\textbf{Return} True
\end{algorithm}
As inputs, we have the system \system and a candidate invariant property $\phi$. At line 2, we set our initial set of predicates to be the set of $\Sigma(I)$-atoms occurring in the initial formula of the system, and in the property $\phi$. 
Then, at line 3, the algorithm checks whether there is a violation of the property in the initial formula. If a counterexample is not encountered, the algorithm is initialized and the main PDR/IC3 loop starts. The algorithm maintains a set of frames $F_0, \dots, F_N$, which are an approximate reachability sequence of $\Tilde{S}_\mathcal{P}$, 
as defined in Definition~\ref{def-approx-reach}. 
Initially, we set $F_0 = \{ \bar{\iota} \}$, $F_1$ to be empty, and we set $N$, the counter of the length of the frame, to be equal to $1$.
Then (line 8), we loop over models in the intersection between the last frame $F_N$ and $\neg \hat{\phi}$. For each of those models, we consider its restriction on the index sort as defined in the preliminaries. Then, the function $diag$ computes the diagram of the (finite) model over the signature $ \Sigma_I \cup X_P$ (Definition \ref{defn-diagram}). Afterwards, the procedure \textit{RecBlock} (Algorithm \ref{recblock}) tries to either block such a diagram, showing that the corresponding models are unreachable from the previous frame with an abstract transition, or computes a diagram from the set of backward reachable (abstract) states, and recursively calls itself. If eventually a diagram is blocked, then corresponding frames are strengthened by adding a (generalization of) the diagrams to them. If the procedure finds a diagram in the first frame, then we are in the presence of an abstract counterexample of $\hat{S}_\mathcal{P}$:
\begin{defn}(Abstract Counterexample)
    Let $F_0, \dots, F_N$ be an \textit{approximate reachability sequence} for $\hat{S}_\mathcal{P}$ and $\hat{\phi}$. An abstract counterexample is a sequence of 
    models $\pi \defeq \sigma_0, \dots, \sigma_N$ such that:
        \begin{itemize}
        \item $\sigma_i \models F_i, \  \forall i. 0 \leq i < N$;
        \item $diag(\sigma_i) \wedge \hat{\tau} \wedge diag'(\sigma_{i+1})$ is satisfiable;
        \item $\sigma_N \models \neg \hat{\phi}$.
    \end{itemize}
\end{defn}
Note the fact that we compute diagrams only over index model and over the signature $\Sigma_I \cup X_P$. This is enough to capture abstract counterexample for $\hat{S}_\mathcal{P}$; we have: 
\begin{prop}
    \label{updr-abstract-cex}
    Given a set of predicates $\mathcal{P}$, if the procedure RecBlock (Algorithm \ref{recblock}) returns false, then there exists an abstract counterexample for $\hat{S}_\mathcal{P}$.
\end{prop}

In line 12, the abstract counterexample is analyzed: if the counterexample is spurious, then we try to refine the set of predicates, and the loop continues. Our refinement procedure is described in the next section. Else, if a concrete counterexample is found, the algorithm terminates with an unsafe result. 
Finally, the propagation phase of the algorithm tries to push diagrams in $F_i$ to $F_{i+1}$. If two frames are equal during this phase, then an inductive invariant is found by the algorithm. Otherwise, the loop restarts with a larger trace.

\begin{exm}
We give an overview of some steps of the procedure over the bakery algorithm. Given $\mathcal{P}(i_1)$ and $X_P$ of the previous examples, we have that
\[ \bar{\iota} = \forall i. \big( x_{state[i] = idle}(i_1) \wedge x_{t[i_1] = 0}(i) \wedge x_{next\_ticket = 1} \wedge x_{to\_serve = 1} \big)\]
and
\[\bar{\phi} = \forall i, j. \big (i \neq j \rightarrow \neg (x_{state[i_1] = crit}(i) \wedge x_{state[i_1] = crit}(j) \big). \]
$H_\mathcal{P}(X_{\mathcal{P}}, X)$ was also shown in the previous example. It is easy to see that 
$\bar{\iota} \wedge H_{P} \wedge \neg \bar{\phi}$ is unsatisfiable. Therefore, we have $F_0 = \hat{\iota}$ and $F_1 = \top$.
We enter the main loop, and consider models of $H_P \wedge \neg \bar{\phi}$; this formula is
satisfiable, and we consider the restriction of a model over the index sort. Suppose that this restriction is given by:

\begin{itemize}
    \item A finite index universe $\{a, b  \}$ with $a \neq b$
    \item $x_{next\_ticket = 1}$ is true in the model, $x_{to\_serve = 1}$ is false; 
    \item the predicate $x_{state[i_1]=idle}$ do not hold for $a, b$ but $x_{state[i_1]=crit}$ does.
\end{itemize}
So, the corresponding diagram over $X_{\mathcal{P}}$ is 
\begin{align*}
    \psi = \exists i_1, i_2. \big( i_1 \neq i_2 \wedge x_{next\_ticket = 1} \wedge \neg x_{to\_serve = 1} \wedge \neg x_{state[i_1]=idle}(i_1) \\ 
    {} \wedge \neg x_{state[i_1]=idle}(i_2) \wedge x_{state[i_1]=crit}(i_1) \wedge x_{state[i_1]=crit}(i_2) \big). 
\end{align*}
We now call $RecBlock(\psi, 1)$, and therefore consider $AbsRelInd(F_0, \tau, \neg \psi, \mathcal{P})$. This is unsatisfiable, so we can add (a generalization of) $\neg \psi$ to $F_1$. Eventually, $F_1 \wedge H_P \wedge \neg \hat{\phi}$ will be no longer satisfiable, so we introduce a new frame $F_2$. (For simplicity, we skip the propagation phase). At this point, $F_2 \wedge H_P \wedge \neg \hat{\phi}$ is satisfiable, and it is possible to have a sequence of models and corresponding diagrams $\psi, \psi', \psi''$ such that $RecBlock(\psi, 2)$ recursively calls $RecBlock(\psi', 1)$ which again calls $RecBlock(\psi'', 0)$. This corresponds to an abstract counterexample.
\qed
\end{exm}

The generalization phase, at line 7 of Algorithm \ref{recblock}, is not relevant to the soundness of the algorithm, but it is a crucial part for its efficiency. During this phase, the diagram $\neg \psi$ is weakened to a more general formula $g$, which implies the negation of the diagram but blocks more models. In our implementation, we used a technique based on unsat core extraction, used also in other PDR/IC3 variants \cite{updr, msatic3}. 

\subsection{Concretizing counterexamples and refinement}
\label{sub-concretizing-updria}
Given an abstract counterexample, we can try to associate it with a concrete counterexample, thus concluding the algorithm with a negative result, or we could discover that no concrete counterexample corresponds to the abstract one. More formally, together with an abstract counterexample $\pi$, the algorithm finds a sequence of diagrams $\psi_0(X_\mathcal{P})$, \dots, $\psi_k(X_\mathcal{P})$. However, differently from the Boolean and quantifier-free case, the abstract unrolling
\begin{equation}
  \psi_0(X_\mathcal{P}^0)\wedge \bigwedge_{i=1, \dots, k} \hat{\tau}(X_\mathcal{P}^{i-1}, X_\mathcal{P}^i) \wedge \psi_i(X_\mathcal{P}^i)
\end{equation}
can still be unsatisfiable. That is, we may have two models $\sigma_1, \sigma_2$ such that $diag(\sigma_1) \wedge \hat{\tau} \wedge diag'(\sigma_2)$ is satisfiable, but there is no abstract transition between $\sigma_1$ and $\sigma_2$. 
A possible reason for this is that $\sigma_1$
and $\sigma_2$ may have different universes, and the diagrams abstract $\sigma_1$ and $\sigma_2$ by upward-closing them w.r.t. the submodel relation  (see section 5 of \cite{updr} for examples and more details). This is also a proof that, to eliminate the abstract counterexample, a universal invariant over $\mathcal{P}$ is not enough. 
The automatic discovery of invariants with quantifier alternation is an active area of research \cite{DistAI, pdh, DBLP:conf/fmcad/GoelS21}; however, in this work, we focus only on universal invariants (note that the frames are the conjunction of negations of diagrams, therefore they are universally quantified in prenex normal form). Therefore, when we encounter an abstract counterexample, we try to expand our predicate set $\mathcal{P}$ and search for a new universal invariant on a larger language. In practice, in order to check whether a sequence of diagrams corresponds to a concrete counterexample, we could check the satisfiability of the concrete unrolling of the system, by replacing atoms in the diagram with their non-abstracted version, i.e. checking if the formula
\begin{equation}
\label{concrete-unrolling}
  \psi_0[X_\mathcal{P}/\mathcal{P}](X^0) \wedge \bigwedge_{i=1, \dots, k} \tau(X^{i-1}, X^i) \wedge \psi_i[X_\mathcal{P}/\mathcal{P}](X^i).  
\end{equation}
is satisfiable. If such a query is satisfiable, then we are given a sequence of models which leads to a violation of the property. In case of unsatisfiability of the (\ref{concrete-unrolling}), a typical approach would be to consider the sequence of formulae $\iota_0 = \psi_0(X^0)[X_\mathcal{P}/\mathcal{P}]$, and $\iota_i = \tau(X^{i-1}, X^i) \wedge \psi_i[X_\mathcal{P}/\mathcal{P}](X^i) $ (for all $1 \leq i \leq n-1 $), and extract an interpolant sequence (\ref{subs-interpol}) from it. However, there are two main practical issues in using directly this encoding: $(i)$ since the formula (\ref{concrete-unrolling}) contains quantifiers, proving its (un)satisfiability is very challenging, and $(ii)$ extracting interpolants from quantified queries is scarcely supported by existing solvers. 

We propose instead to consider an under-approximation of (\ref{concrete-unrolling}), i.e. to consider a finite model encoding of it in a certain size. That is to say, instead of trying to unroll the counterexample in all the instances of the system, we try to block it in a ground instance $S_n$.  To choose an appropriate size, recall that a diagram $\psi_i$ is an existentially closed formula that is built from an index model with cardinality $n_i$, where $n_i$ is the number of existentially quantified variables in $\psi_i$. Thus, we need to choose a size able to represent at least all the diagrams in the abstract counterexample. Therefore, we take $n = max\{n_i | i=1, \dots, k\}$, which is the least integer such that each diagram is satisfiable (we take the least for tractability reasons). Then, we consider  
\begin{equation}
\label{concrete-unrolling-finite}
  (\psi_0[X_\mathcal{P}/\mathcal{P}])_n(C^0, X^0) \wedge \bigwedge_{i=1, \dots, k} \tau_n(C^i, X^{i-1}, X^i) \wedge (\psi_i[X_\mathcal{P}/\mathcal{P}])_n(C^i ,X^i).  
\end{equation}
where $C$ and $\tau_n$ are defined in \ref{subs-ground}. If such under-approximated unrolling is satisfiable, we have a counterexample in $S_n$, and we can terminate the algorithm with an unsafe result. Otherwise, we can compute an interpolant sequence, and add all the atoms of the interpolants (after replacing the symbols in $C$ with free variables) to the predicate set $\mathcal{P}$. Note however that such a refinement will rule out (\ref{concrete-unrolling-finite}), but in general not guaranteed to rule out (\ref{concrete-unrolling}): that is, the counterexample can occur again in a greater size. If this happens, we will consider larger and larger finite encodings, and repeat the process. In this way, such a refinement will diverge when an invariant with not only universal quantification is needed, as the formula \ref{concrete-unrolling} remains sat, and larger and larger ground instances will be explored, always yielding the same predicates.

\begin{exm}
    Continuing the latter example, suppose we are given a sequence of diagrams $\psi'', \psi', \psi$, each with two existentially-quantified variables. Therefore, we consider the unrolling of (\ref{concrete-unrolling-finite}) in $S_2$, by introducing two fresh index constants $c_1$, $c_2$. Such a formula is unsatisfiable, yielding an interpolant $t[c_1] = to\_serve$. Therefore, we restart the loop with $\mathcal{P} = \mathcal{P} \cup \{t[i_1] = to\_serve \}$.
    \qed
\end{exm}
 
\subsection{Properties}

We have the following properties of the algorithm:
\begin{prop}(Soundness)
\label{prop-updria-soundness}
    If the algorithm terminates with safe, then $S \models \phi$. If the algorithm terminates with unsafe, then $S \not \models \phi$.
\end{prop}
Moreover, given propositions \ref{prop-rel-inductive-1} and \ref{updr-abstract-cex}, we can establish the following fact about Algorithm \ref{updralg}: once a set of indexed predicates $\mathcal{P}$ is fixed, our procedure simulates the standard UPDR \cite{updr} algorithm over the system $\hat{S}_\mathcal{P}$. As already anticipated, the algorithm has a result of partial completeness over a certain set of universally quantified formulae. To be more precise, let $\mathcal{L}$ be the set of formulae of the form $\forall \underline{i}. \psi(\underline{i})$, with $\psi(\underline{i})$ a Boolean combination of $\Sigma_I(\underline{i})$-atoms and  elements of $\mathcal{P}(\underline{i})$. We have the following proposition (by lifting Proposition 5.6 in \cite{updr}):
\begin{prop}(Partial Completeness)
\label{completeness}
    If the algorithm \ref{updralg} finds an abstract counterexample, then no inductive invariant in $\mathcal{L}$ exists for $S$ and $\phi$.
\end{prop}

We already mentioned that the problem we are addressing in this paper is undecidable, and there are indeed many possible causes of non-termination of this algorithm. First, we rely continuously on first-order reasoning, which is in general undecidable. However, a nice property of our procedure is the following, which follows from simple logical manipulation after noticing that the formulae defining the implicit abstraction contain only universal quantification:
\begin{prop}
\label{algepr}
    If the input system \system is defined by formulae 
    falling in the decidable fragment of Proposition~\ref{prop-decidable-epr}, then each satisfiability check of Algorithm \ref{updralg} falls again in the decidable fragment.
\end{prop}

Proposition \ref{algepr} ensures that, if the initial and the transition formulae are described in a decidable setting, then we are not stuck in first-order reasoning.  
However, the whole procedure may still be non-terminating: the sources of divergence may be an infinite series of refinements, or the unbounded exploration of an infinite-search space (in fact, even after fixing a set of predicates $\mathcal{P}$, the abstract system $\hat{S}_\mathcal{P}$ is still infinite-state: every model is finite, but there is no bound on the cardinalities of their universes). 

\section{Learning lemmas from ground instances}
\label{sec-lambda}
In the verification of parameterized systems, the exploration of ground instances has always been recognized as a source of helpful heuristics
~\cite{InvisibleInv, DBLP:conf/fmcad/ConchonGKMZ13}. 
The intuition is that in most cases if a counterexample to a property exists, it can be detected for small values of the parameter. Moreover, if a property holds, the reason for that should be the same for all values of the parameter (at least after a certain threshold value).
In this section, we present a second algorithm for solving the invariant problem for array-based transition systems, inspired by the 
ideas above: we try to generalize an invariant for the system by exploring small ground instances. In contrast to the algorithm of the previous section, this algorithm tries to guess an invariant, rather than constructing one incrementally. Even if we do not provide any theoretical guarantees that this guess will eventually be correct, this approach relies less on quantified reasoning, which turns out to be better from a practical standpoint. We start by giving a high-level overview of our method, depicted also in Fig.~\ref{Overall}. Then, we will go into the details of the algorithm, discussing our generalization technique \ref{subs-gener}, and presenting two different approaches to check whether our generalization is correct (the Candidate Checking box) \ref{subs-param-abs}, \ref{subs-smt-solving}. We finish the section by discussing the properties of the algorithm \ref{subs-properties-2}.

\subsection{High level algorithm}
\label{sec-overview}
\begin{figure}[ht]
\centering
 \scalebox{.7}{
\begin{tikzpicture}[node distance=1.7cm]
\node (in1) [io] {$S$, $\phi$};
\node (lem) [lemmas, below of=in1, yshift=-0.5cm]{$\Psi$};
\node (ground) [process, right of=in1, xshift=1cm]{$S_n$ Ground instance};
\node (check) [process, right of=ground, xshift=2cm]{$S_n \models \phi_n \wedge \Psi_n$?};
\node (spurious)  [process, right of=check, xshift=3.5cm]{Does $\pi \models \neg \phi_n$?};
\node (gen) [process, below of=check, yshift=-1cm]{Generalization};
\node (inv) [process, below of=gen, yshift=-1cm]{Candidate Checking};
\node (end2) [io, left of=inv, xshift=-2cm]{\Safe};
\node (end1) [io, right of=spurious, xshift=2cm] {\Unsafe};
\node (badlem) [process, below of=spurious, yshift=-1cm]{Remove $\psi_i$ such that $\pi \models \neg \psi_i$};

\draw [arrow] (in1) -- (ground);
\draw [arrow] (lem) -- (ground);
\draw [arrow] (ground) -- (check);
\draw [arrow] (check)  -- node [anchor=south]{No, $\pi$ cex}(spurious);
\draw [arrow] (spurious) -- node [anchor=south]{Yes}(end1);
\draw [arrow] (spurious) -- node [anchor=east]{No}(badlem);
\draw [arrow] (badlem) -- (check);
\draw [arrow] (check) -- node [anchor=east, align=center] {Yes,\\$Inv$ invariant}(gen);
\draw [arrow] (gen) -- (inv);
\draw [arrow] (inv) -- node [anchor=south]{Yes}(end2);
\draw [arrow, bend left, align=center] (inv) to node [anchor=east]{No,\\$n=n+1$}(ground);

\end{tikzpicture}}
 \caption{An overview of the algorithm.\label{Overall}}
\end{figure}
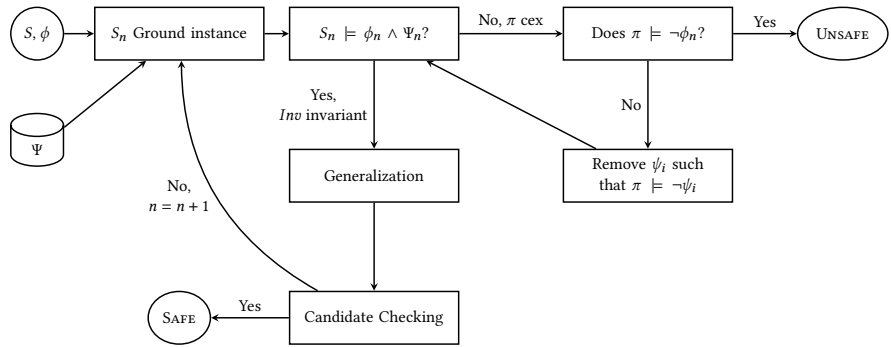

We consider as inputs for the algorithm an array-based transition system $S$ and a candidate invariant $\phi$. 
We also maintain a set of formulae $\Psi = \{\psi_1, \dots, \}$, initially empty, that contains universally quantified
 \textit{candidate lemmas}. We also initialize $n$, a counter for the size of the ground instance we explore, equal to 1. We perform the following steps:
\begin{itemize}
\item we consider $S_n$, as defined in \ref{subs-ground}, and then use a model checker to check whether $\phi_n \wedge \Psi_n$ is an invariant for the ground instance. If we get a counterexample, we check if the property itself is falsified
(thus terminating the algorithm with \Unsafe result), or we remove the lemmas that are proved to be false in size $n$.
In the latter case, we repeat the ground model checking query, until either a true counterexample is obtained, or the property holds.

\item If the model checker proves the property, we consider an inductive invariant $Inv$ of $S_n$. From such a formula, we synthesize a set of new lemmas, that we add to $\Psi$ (Generalization).%

\item In the Candidate Checking box, we check if the property $\phi$ (together with the lemmas $\Psi$) is an inductive invariant for the original system $S$. We propose two methods for performing this check, described in \ref{subs-param-abs}, \ref{subs-smt-solving}. 

\item In case of success, the property is proved and we have found an inductive invariant.  In case of a failure, we need a better candidate invariant: we restart the loop with a new exploration from size $n+1$. 
\end{itemize}
\noindent

Note there are many possible causes of non-termination of the algorithm (again, since we are dealing with undecidable problems, this is not avoidable): the main problems are the candidate checking box, which involves quantified reasoning, and the existence of a cut-off, i.e. an integer $n$ such that the generalized formula obtained 
after model checking a ground instance of size $n$ is inductive also for all other instances. 
The pseudocode of the procedure is reported in Algorithm \ref{alg-gener}. In the next sections, we describe all the sub-procedures in detail. 

\begin{algorithm}[ht]
\SetAlgoLined
\caption{Learning lemmas from ground instances}
\label{alg-gener}
Input: $S = (X, \iota(X), \tau(X, X')), \phi(X)$
\\
$n = 1$
\\
\textbf{while} True: \\
\Indp $S_n = grounding(S, n)$\\
\textbf{while} not $S_n \models \phi_n \wedge \Psi_n$: \\
\Indp \textit{\# model checker returns a cex $\pi$ from size $n$} \\
\textbf{if} $\pi \models \neg \phi_n$:
\\
\Indp \textbf{return} \Unsafe, $\pi$ 
\\\Indm
\textbf{else:}\\
 \Indp $\Psi = \Psi \setminus \{\psi_i | \pi \models \neg \psi_i \}$ \\
\Indm \Indm \textit{\# model checker returns an inductive invariant $Inv$ from size $n$}\\
$\Psi = \Psi \cup generalization(Inv)$\\
\textbf{if} InvariantChecking($S, \phi, \Psi$):
\\
\Indp \textbf{return} \Safe, $\phi \wedge \Psi$ 
\\
\Indm \textbf{else:}\\
\Indp  \textbf{$n=n+1$}
\end{algorithm}

\subsection{Generalization}\label{subs-gener}
\newcommand{\Alldiff}{\ensuremath{\textsf{\textit{AllDiff}}}\xspace}
A key step of this algorithm is the one of producing candidate invariants by generalizing proofs from a ground instance. In this section, we show our approach. We suppose to have a model checker capable of proving or disproving that $S_n \models \phi_n$ (e.g. \cite{msatic3}). If a counterexample is not found, we also suppose that the model checker returns a formula $Inv$ which witnesses the proof, i.e. an inductive invariant of $S_n$ and $\phi_n$. From this witness, we generalize a set of quantified lemmas. 
\begin{defn}[Generalization]
\label{def-general}
Let $S$ be an array-based transition system, 
$\phi$ a candidate invariant, and suppose $S_n \models \phi_n$. 
A $n$-\emph{generalization} is a (quantified) formula $\Psi$ 
such that $\Psi_n$ is an inductive invariant for $S_n$ and $\Psi_n \models \phi_n$.
\end{defn}
In practice, we exploit the same technique that we used in~\cite{Lambda}, inspired by~\cite{InvisibleInv}, which works as follows. Suppose that $Inv$
is in CNF. Then, $Inv = \mathcal{C}_1 \wedge
\dots \wedge \mathcal{C}_m$ is a conjunction of clauses. From every one of such clauses, we will obtain a universally quantified formula.
Let $\Alldiff(I)$ be the formula that states that all variables in $I$ are different from each other. 
For all $j \in \{1, \dots, m \}$, let $\psi_j = \forall I.\Alldiff(I)\rightarrow \mathcal{C}_j [C / I]$ - where $C$ are the variables introduced in \ref{subs-ground}. 
Let $\Psi = \bigwedge_{j=1}^m \psi_j $. It follows from Proposition (\ref{symm}) that such a $\Psi$ is a $n$-generalization.

It should be clear that our technique can infer lemmas with only universal quantifiers, but more generalizations are possible~\cite{ic3po, DBLP:conf/cav/DooleyS16}. For example, if a clause of the inductive invariant were $l(c_1) \vee ... \vee l(c_n)$, a possible generalization of that clause would be $\exists x. l(x)$.

\subsubsection{Symmetries of ground instances}
\label{subs-symm}
As already observed in previous works~\cite{Krstic2005ParametrizedSV,ic3po,Lambda}, 
transition systems obtained by instantiating quantified formulae have a certain degree of symmetry. 
We report here the notion that will be useful to our description. Let $\sigma$ be a permutation of ${1, ..., n}$ (also called $n$-permutation), and $\phi$ is a formula in which $c_1, \dots, c_n$ occur free. We denote with $\sigma \phi$ the formula obtained by substituting every $c_i$ with $c_{\sigma(i)}$. 
The following proposition follows directly from the fact that $\iota_n$ and  $\tau_n$ are obtained by instantiating a quantified formula with a set of fresh constants $C$\cite{Krstic2005ParametrizedSV}:
\begin{lem}
For every permutation $\sigma$, we have that:
$(i)$ \tarray $\models \sigma \iota_n \leftrightarrow \iota_n$; $(ii)$ 
\tarray $\models \sigma \tau_n \leftrightarrow \tau_n$.
\end{lem}
From this lemma and an induction proof, the following holds:
\begin{prop}[Invariance for permutation]
\label{symm}
Let $s$ be a state of $S_{n}$, reachable in $k$ steps. Then $s \models \phi(C, X)$, if and only if, for every $n$-permutation $\sigma$, there exists a state $s'$ reachable in $k$ steps such that  $s' \models \sigma \phi(C, X)$
\end{prop}
We exploit this property both for the verification of ground instances and in the generalization process.  In fact, from the last proposition, we can simplify every invariant problem $S_{n} \models \bigwedge_\sigma \sigma \phi(C, X)
$ -- where $\sigma$ ranges over all possible substitutions -- to $S_{n} \models \phi(C, X)$. This simplification is of great help when checking properties which are the result of instantiating a formula with only universal quantifiers. 
\subsubsection{Minimizing modulo symmetries}
\label{subs-minimizing}
Our next step will be to try to prove that the generalized lemmas from size $n$  also hold for all other ground instances.
Therefore, it is intuitive to try to weaken as much as possible the generalization $\Psi$, to increase the chances that its inductiveness will be preserved in other instances. 
So, before generalization, we exploit the invariant minimization techniques described in~\cite{FEAS} to weaken the inductive invariant $Inv$ by removing unnecessary clauses. 
However, note that, with our generalization technique, two symmetric clauses produce the same quantified formula: if $\sigma$ is a substitution of the $C$'s, 
the formulae obtained by generalizing a clause $\mathcal{C}(C)$ or $\sigma \mathcal{C}(C)$ are logically equivalent. 
So, we apply the following strategy: given a clause $\mathcal{C}(C)$ in $Inv$,
we add to the invariant all the `symmetric' versions $\sigma \mathcal{C}$, where $\sigma$ ranges over all possible substitutions of the $C$'s. By Proposition (\ref{symm}), we can safely add 
those clauses to $Inv$ and it will remain inductive. Then, during the minimization process, a clause is removed from the invariant only if all its `symmetric'
 versions are. 

\subsection{Candidate Checking via Parameter Abstraction}
\label{subs-param-abs}
In this section, we show how we can check that a formula $\phi$ is an invariant of $S$ by using only a quantifier-free model checker.
We do this by constructing a new system $\Tilde{S}$, called the parameter abstraction of $S$, which is quantifier-free and simulates $S$. The first version of this approach was introduced in \cite{Krstic2005ParametrizedSV}; in the following, we describe a novel version of the abstraction, and how it can be applied to array-based transition systems. The main novelty is that, instead of using a special abstract index ``$*$'' that over-approximates the behavior of the system in the array locations that are not explicitly tracked, we use $n$ \emph{environmental (index) variables} which are allowed to change non-deterministically in some transitions. This can be achieved by the usage of an additional stuttering transition: this rule allows the environmental variables to change value arbitrarily, 
while not changing the values of the array in all other indexes.

\textbf{Hypothesis}: In the remaining of this section, we suppose an additional hypothesis on the shape of the system $S$. These hypotheses are needed for the construction of the abstraction system $\Tilde{S}$, and they turn out not to be particularly restrictive: for instance, large classes of parameterized systems can still be described in this setting. In particular, we suppose that the candidate property $\phi(X)$ and the initial formula $\iota(X)$ are universal formulae. Moreover, we suppose that the transition formula $\tau(X,X')$ is a disjunction of formulae of the form $\exists I \forall J. \psi(I, J, X, X')$, with $\psi$ quantifier-free.

\subsubsection{Preprocessing steps}
Before starting constructing the abstract system, 
we apply some preprocessing steps to $S$, justified by the following two proposition. The first proposition is basically an induction principle:
\begin{prop}[Guard strengthening~\protect{\cite{Krstic2005ParametrizedSV}}]
\label{thm-guard-strengthening}
Let $C = (X,\iota(X),\tau(X,X'))$ be a symbolic transition system and let $\phi(X)$ be a formula.
Let ${C_{\phi}}$$ \defeq (X, \iota(X),$ $\tau(X,X') \wedge\phi(X) )$ (called the guard-strengthening of $C$ with respect to $\phi(X)$).
If $\phi$ is an invariant of $C_{\phi}$, then it is also an invariant of $C$.
\end{prop}
The second proposition is used to remove quantified variables in candidate properties:
\begin{prop}[Removing quantifiers\protect{\cite{EagerAbs}}]
\label{thm-prophecy-instantiation}
  \label{prophecy}
Let $C = (X, \iota(X), \tau(X, X'))$ be a symbolic transition system. 
Let $\forall X. \phi(X)$ be a formula, and $P$ a set of constants of the same length of $X$. Then, $\forall X. \phi(X)$ is an invariant for $C$ iff the formula $\phi[X/P]$  is an invariant for $C_{P}=(X\cup P, \iota(X), \tau(X, X')\wedge P'=P)$, where $P$ is a set of fresh frozen variables, called \textit{prophecy variables}.
\end{prop}

In order to build the parameter abstraction of the system $S$, we first apply Proposition \ref{thm-guard-strengthening}, and 
add the candidate property itself in conjunction with the 
transition formula. Then, by applying Proposition 
\ref{thm-prophecy-instantiation}, we remove the 
universal quantifiers in the formula $\phi$ and 
introduce accordingly a set of prophecy variables $P$.
In this way, we reduce the problem of checking $S \models \forall I. \phi'(I, X)$ 
to $S \models \phi'(P, X)$.

\subsubsection{Abstraction Computation}\label{sec-abstraction-defn}

We continue the computation of the abstract system by defining a set of fresh variables, called the environmental variables $E$, in a number determined by the greatest existential quantification depth in the disjuncts of the transition formula of $S$. While the prophecies are frozen variables, 
the interpretation of the environmental variables is not fixed. Moreover, we assume that the values taken by $P$ and $E$ are different.
We now define the formulae for $\tilde{S}$, the parameter abstraction of $S$.
\subsubsection*{Initial formula}
Let $\iota(X)$ be the initial formula of $S$, which by hypothesis contains only universal quantifiers.  The initial formula of the abstract system is a formula $\tilde{\iota}(P, X)$, obtained by expanding all the universal quantifiers in $\iota$ over the set of prophecy variables $P$.

\subsubsection*{Transition formula}

For simplicity, we can assume that we have only one disjunct in $\tau(X,X')$. First, we compute the set of all substitutions of the existentially quantified variables $I$ over $P\cup E$, and we consider the set of formulae $\{\tilde{\tau}_j(P, P', E, X, X')\}$, where $j$ ranges over the substitutions, and $\tilde{\tau}_j$ is the result of applying the substitution to $\tau$.

Then, for each formula in the set $\{\tilde{\tau}_j\}$, we instantiate the universal quantifiers in it over the set $P \cup E$, obtaining a quantifier-free formula over prophecy and environmental variables.

Moreover, we consider an additional transition formula, called the \textbf{stuttering transition}, defined by:
\[\tilde{\tau}_S \defeq \bigwedge_{x \in X} \bigwedge_{p \in P} x'[p] = x[p] \wedge p' = p \]

The disjunction of all the abstracted transition formulae is the transition formula $\tilde{\tau}$. So, we can now define the transition system
$$
\tilde{S}
\defeq
(\{X, P, E\}, \tilde{\iota}(P, X), \tilde{\tau}(P, P', E, X, X')).
$$


\subsubsection{Stuttering Simulation}
We state here the property of our version of the Parameter Abstraction,  the proof of which can be found in the appendix. Formally, the abstraction induces a stuttering simulation, where the stuttering is given by $\tilde{\tau}_S$:  this is a weaker version than the simulation induced by \cite{Krstic2005ParametrizedSV}, yet it is sufficient for preserving invariants.
  %
  %
  %
%
Intuitively, in our abstraction, we require that every abstract array is equal to its concrete counterparts only on the locations referred to by the prophecy variables. We then have the following:

\begin{prop}
\label{prop-general-simulation}
There exists a stuttering simulation between $S$ and $\tilde{S}$. Moreover, if $\tilde{S} \models \Phi(P, X)$, then $S \models \Phi(P, X)$.
\end{prop}

Unfortunately, if $\tilde{S} \not \models \Phi(P, X)$, we cannot conclude anything, since the counterexample may be spurious, due to a too coarse abstraction. If that is the case, we restart the loop of the algorithm by increasing the size of ground instance $n$, and either find a counterexample in a larger size or consider again the abstraction with a new candidate invariant.

\subsection{Candidate Checking via SMT solving}
\label{subs-smt-solving}
An alternative -- and more standard -- approach for performing the Candidate Checking procedure, which does not 
require particular syntactical restrictions, 
is to check directly if the formula $\Psi$ is an \emph{invariant strengthening} for $\phi$.

\begin{defn}
Let  $S = (X, \iota(X), \tau(X, X'))$ be an array-based transitions transition system, and $\phi$ a candidate invariant. An invariant strengthening $\Psi$ is a formula such that the following formulae are \tarray -unsatisfiable:

\begin{align}
\label{inductive}
\begin{split}
 \iota(X) \wedge \neg (\phi(X) \wedge \Psi(X)) \\
 \tau(X, X') \wedge 
\phi(X)\wedge \Psi(X) \wedge \neg (\phi(X') \wedge \Psi(X')).
\end{split}
\end{align}

\end{defn}

Since a formula is valid iff its negation is unsatisfiable, it follows from the definition that if $\Psi$ is an invariant strengthening for $\phi$ iff $\phi \wedge \Psi$ is an inductive invariant for $S$.

Checking the unsatisfiability of the formulae \eqref{inductive} could be implemented with the usage of any prover supporting SMT reasoning and quantifiers, e.g. \cite{z3, cvc5}. However,  
especially for satisfiable instances, such solvers can diverge easily. Thus, since many queries can be SAT, a naive usage 
of such tools will cause the procedure to get stuck in quantified reasoning with no progress obtained.
Therefore, we propose in this section a `bounded' sub-procedure of Candidate Checking, in which instead of relying on an off-the-shelf SMT solver supporting quantifiers, we `manually' apply standard instantiation-based techniques for quantified SMT reasoning~\cite{simplify}, in which however we carefully manage the set of terms used to instantiate the quantifiers, in order to prevent divergence.

Given a candidate inductive invariant, we perform Skolemization on the inductive query (\ref{inductive}), obtaining a universal formula. Then, we look for a set of terms $G$ 
such that the ground formula obtained by instantiating the universals with $G$ is unsatisfiable. This is the standard approach used in SMT solvers for detecting unsatisfiability of 
quantified formulae~\cite{simplify,z3mbqa}; the main difference is that instead of relying on heuristics 
to perform the instantiation lazily during the SMT search (e.g \cite{simplify,z3mbqa}), we carefully control the quantifier instantiation procedure, and expand the quantifiers eagerly so that we can use only quantifier-free SMT reasoning.

Let $\phi_S = \forall I. \phi'_S(I, X)$ be the result of the Skolemization process, 
where $\phi'_S$ is a quantifier-free formula over a signature $\Sigma'$, obtained by expanding $\Sigma$ with new Skolem symbols. Initially, we simply let $G$ be the set of 0-ary symbols of the index sort in the formula. 
Note that apart from constants in the original signature, new (Skolem) constants arise by eliminating existential quantifiers. 
Since we use only universal quantification for the generalized invariant strengthening, $\Psi$  is a conjunction of universal formulae, 
and we can swap the conjunction and the universal quantification to obtain a formula with only $n$ universal quantified variables, where $n$ is the size of the last ground instance visited. Notice that,  
since the candidate inductive strengthening occurs also negated in the quantified formula, this will produce $n$ new Skolem constants.%

Finally, we can add to the inductive query an additional constraint. By induction on the structure of our algorithm, if $\Psi$ is generalized from size $n$, we have proven already that the property $\phi$ holds in $S$ for all the ground instances of size equal to or less than $n$. Thus, we impose that in our universe $G$ there are at least $n$ different terms. 

To sum up, let $\phi_S = \forall I. \phi'_S(I, X)$ be the universal formula obtained after Skolemizing the formulae in (\ref{inductive}), and let $m$ be the length of $I$. Let $n$ be the cardinality of the last visited ground instance. Let $G$ be the set of constants of index sort in $\phi'_S$ (by the previous discussion, $|G| \geq n$). Let $c_1, \dots, c_n$ be a set of fresh variables of index sort. We test with an SMT solver the satisfiability of the following formula

\begin{equation}
  \label{inst-check}
\bigwedge_{\underline{g} \in G^m} \phi'_S[I/\underline{g}] \wedge \Alldiff(C) \wedge \bigwedge_{j=1}^n(\bigvee_{g \in G} c_j = g) 
\end{equation}
 \noindent
We have that:
 \begin{prop}
 \label{prop-instantiations}
For any set of $\Sigma_I$-terms $G$, if \eqref{inst-check} is unsatisfiable, then $\psi$ is an inductive strengthening for $\phi$.
 \end{prop}

\subsubsection{Refinement}
If the former formula is SAT, there are two possibilities. Either we have a real counterexample to induction, and we need a better candidate, or our instantiation set $G$ was too small to detect unsatisfiability. 
In general, if $G$ covers all possible $\Sigma_I$-terms, then we can deduce that the counterexample is not spurious. 
\begin{defn}
Given an index theory $\mathcal{T}_I$ with signature $\Sigma_I$, we say that a set of $\Sigma_I$-terms $G$ is saturated if, for all terms $\Sigma_I$-term $t$, there exists a $g \in G$ such that $\mathcal{T}_I \models t = g$. 
\end{defn}
So, if $G$ is saturated, any model of \eqref{inst-check} corresponds to
a counterexample to induction, and we need a better strengthening. However, in case \eqref{inst-check}
is satisfiable, but $G$ is not saturated, we use the following heuristic to decide whether we need a better candidate or a larger $G$. We consider the inductive 
query in $S^{n+1}$,%
using as a candidate inductive invariant $(\Psi \wedge \phi)^{n+1}$. If the candidate invariant is still good (the query is UNSAT), we try to increase $G$ to get the unsatisfiability of the unbounded case.
Our choice is to add to $G$ terms of the form $f(x)$ where $f$ is a function symbol of index type, and $x$ are constants already in $G$. 
Note that if no function symbols are available, i.e. if $\Sigma_I$ is a relational signature, then saturation of $G$ follows already by considering 0-ary terms.
Therefore, in case $G$ is initially not saturated, the existence of at least one function symbol is guaranteed.

If the query (\ref{inst-check}) is now UNSAT, we have succeeded. 
Otherwise, we continue to add terms to $G$, until either all function symbols have been used, or an UNSAT result is encountered. 
If the candidate invariant strengthening is not inductive for size $n+1$ (the query is SAT), then we remove the bad lemmas, and we fail to prove the property.

An important remark is necessary to put more insight on the reasons
why our instantiation procedure is effective, especially for the benchmarks we considered. 
As already mentioned, in many systems descriptions, especially the ones arising from parameterized verification, the formulae describing array-based systems fall into the decidable fragment of Proposition~\ref{prop-decidable-epr}. 
In this case, no function symbols are introduced during Skolemization: therefore, the set $G$ of 0-ary terms already is saturated.
Even in the case of  $\forall \exists$ alternation (but in a multi-sorted setting), saturation can be achieved after a few refinement steps 
(as long as the Skolem functions introduced in the signature do not combine in cycles). More details about the completeness of instantiation methods, especially for the verification of parameterized
systems, can be found in~\cite{MCMT, DBLP:journals/lmcs/FeldmanPISS19}. 
Since we limit ourselves to terms of depth one, our method can fail to prove invariants requiring some more complex 
instantiations. Note that in that case, it is always possible to change the choice and the refinement of the set $G$ with more sophisticated methods~\cite{ematch,z3mbqa}.

\subsection{Properties}
In this section, we state the general properties of the algorithm, that hold regardless of the strategies used during the procedure.
\label{subs-properties-2}
\begin{prop}(Soundness)
    If Algorithm~\ref{alg-gener} returns \Safe, then $S \models \phi$. If the algorithm returns \Unsafe, then $S \not \models \phi$.  
\end{prop}
\begin{proof}
    The first claim follows directly from propositions \ref{prop-general-simulation} or \ref{prop-instantiations}. For the second claim, the algorithm terminates with unsafe only when $S^n \not \models \phi$, meaning that there exists a model 
    $\mathcal{M}$ of \tarray{} of cardinality $n$, and a sequence of states $s_0, \dots, s_k$, such that $\mathcal{M}, s_0 \models \iota$ and $\mathcal{M}, s_i, s_{i+1} \models \tau$ for all $0 \leq i < k$. This also means that $S \not \models \phi$.
\end{proof}
In general, we do not have any theoretical guarantees that our algorithm will terminate. In fact, even the 
ground Candidate Checking  $S_n \models \phi_n \wedge \Psi_n$ can be non-terminating. However, our algorithm guarantees
 (it follows from how we do the generalization, i.e. Definition~\ref{def-general}) that every lemma in $\Psi$ can be removed only after checking an $n'$ instance, with $n' > n$. Therefore, the algorithm always makes progress in the following sense.
\begin{prop}(Progress)
    During every execution of the loop of Algorithm \ref{alg-gener}, the same pair $(n, \Psi)$ never occurs twice.
\end{prop}
Finally, by checking instances of bigger and bigger size, we semi-decide the problem of falsifying invariant problems:
\begin{prop}(Semi-completeness for counterexamples)
    Suppose $S \not \models \phi$, and the minimal counterexample wrt of the sizes of index models is in a certain size $n$. If all the model checking problems $S_{n'} \models \phi_{n'}$, with $n' < n$, terminate, 
    then algorithm \ref{alg-gener} eventually finds the counterexample. 
\end{prop}

\section{Related Work}
\label{sec-rel-work}

Verification of systems with quantifiers ranging over finite but unbounded domains has always received a lot of attention from the literature. In particular, a main area of application, for which our method is designed, is parameterized verification. 

In this field, various techniques are based on \textit{cut-off} results. In our terms, a cut-off is a size of a ground instance that contains already all possible behaviors. Cut-off values exist for large varieties of classes of systems, but such results strongly depend on assumptions such as topology, data, etc (see \cite{ParamDecidab} for a survey). In contrast, our methodology aims to be deductive and more general. Still, the idea of generalizing from ground instances is the main block of the algorithm described in Section \ref{sec-lambda}. Indeed, we can see \textit{a posteriori} that, when the algorithm terminates, we have generalized the behaviors of a ground instance to the parameterized case.

The approach of invisible invariants~\cite{InvisibleInv} was to our knowledge the first which proposed the usage of ground exploration to produce candidate universally quantified invariants. In that paper, a candidate invariant is generalized from the formula describing the set of reachable states of a ground instance. The technique was BDD-based and considered only Boolean systems. Therefore, systems considered in that work are a subclass of ours.

Abstraction methods are another major trend in the verification of quantified systems.  Within this family of abstractions, earlier versions of the Parameter Abstraction \cite{Krstic2005ParametrizedSV,SimpleMethod} have been used
successfully also for industrial protocols \cite{Talupur}. 
The main drawback is that the degree of automation is limited, and substantial expertise is required to obtain the desired results. The first steps of our abstraction algorithm of Section \ref{subs-param-abs} are inspired by the ones in \cite{EagerAbs} and \cite{Krstic2005ParametrizedSV}. 

Tools designed for the automatic verification of systems with a combination of first-order quantifiers and SMT theories are \mcmt \cite{MCMT} and \cubicle \cite{cubicle}. These tools use the framework of array-based transition systems but with additional restrictions on the shape of the formulae defining the system. For example, all the assignments of the next state variables in the transition formulae must be functional, i.e. only deterministic updates are allowed. Both tools implement a fully symbolic backward reachability algorithm, where pre-images of states can be described by symbolic quantified formulae. Quantified queries are then discharged with an instantiation approach similar to Section \ref{subs-smt-solving}. Nonetheless, some approximations can be introduced during the backward computation, which may cause spurious counterexamples\cite{alberti2012universal}. \cubicle extends this algorithm by using finite instance exploration to speed up pre-image computation \cite{DBLP:conf/fmcad/ConchonGKMZ13}. Instead, our approaches are not based on the computation of pre-images.  

Ivy \cite{Ivy2016, DBLP:journals/lmcs/FeldmanPISS19} is a tool for the verification of inductive invariants of a class of systems, which can be embedded in our formalism. The tool is not completely automatic and guides the user to write manually inductive invariant. Inspired by Ivy, \mypyvy \cite{QfSep} is a tool that implements algorithms for the automatic discovery of inductive invariants. Among those we have UPDR\cite{updr}, the mentioned version of the IC3 algorithm capable of inferring universally quantified invariants, that we extended in Section \ref{sec-updria} to handle a general theory $\mathcal{T}_E$. Other algorithms available are FOL-IC3 \cite{QfSep}, which extends IC3 by using separators to find invariants with quantifier alternation during the construction of frames, and PdH, a recent algorithm that combines the duality between states and predicates to discover invariants  \cite{pdh}.

Additionally, there are various tools in the literature are designed to use finite instance exploration to guess invariants to lift to the unbounded case \cite{ic3po, i4, DistAI, DBLP:conf/nsdi/HanceHMP21, AutomApproach}. These tools either rely on cut-off results, and so are domain-specific, and on some external prover to discharge the quantified queries. In particular, the tool \icrpo~\cite{ic3po} performs an algorithm similar to the one in Section \ref{sec-lambda}.  Moreover, it implements a generalization technique that can also infer formulae with quantifier alternation by detecting symmetries in a clausal proof. All these methods however are not applicable in the presence of a general SMT theory $\mathcal{T}_E$, but only in the case of pure first-order logic.

Another SMT-based approach for parametric verification is in \cite{SMTparam}. The method is based on a reduction of invariant checking to the satisfiability of non-linear Constrained Horn Clauses (CHCs). Besides differing substantially in the overall approach, the method is more restrictive in the input language and handles invariants only with a specific syntactic structure.

IC3 with implicit abstraction is presented in \cite{msatic3} as an efficient combination of Boolean IC3 and predicate abstraction, which does not rely on expensive ALLSAT procedures for the computation of the abstract system. The algorithm in Section \ref{sec-updria} follows the same ideas, but instead combines UPDR and indexed predicate abstraction \cite{indexpredicates} to obtain a general algorithm for array-based transition systems. 

The use of prophecy variables for inferring universally quantified invariants has also been explored in other contexts, such as \cite{tacas21, DBLP:conf/vmcai/VickM23}. The main difference with our work is that they focus on finding quantified invariants for quantifier-free transition systems with arrays, 
rather than array-based systems with quantifiers. The overall abstraction-refinement approach is also substantially different.


\section{Experimental Evaluation}
\label{sec-expval}
We organize the experimental evaluation as follows: first, we compare the two algorithms and discuss different options in Section \ref{subsec-comparison-algos};
then, we compare the algorithms to other tools in Section \ref{subsec-comparison-others}.
The tools we considered are the model checkers \mcmt, \cubicle, \icrpo, and \mypyvy, all previously mentioned in Section \ref{sec-rel-work}. 
All benchmarks and our implementations are available at the following link: \url{https://drive.google.com/file/d/1_lIUa_Y-yKAhj5bXnX1hLEovFmP3Jos9/view?usp=sharing}. 

We have run our experiments on a cluster of machines with a 2.90GHz Intel Xeon Gold 6226R CPU running Ubuntu Linux 20.04.1, using a time limit of 1 hour and a memory limit of 4GB for each instance. 

\subsection{Implementation}
We have implemented the algorithms described in sections \ref{sec-updria} and \ref{sec-lambda} in Python. 
In this section, we test the two algorithms on various benchmarks, by considering different options and back-end engines. 
For the first algorithm, which we will denote in the following as \UPDRIA, we used the SMT solver \zr \cite{z3} for satisfiability queries (using the approach \cite{z3mbqa} for quantifiers). We also implemented a simple loop to ensure that every index model found by the solver would be of minimal (finite) size: in theory, when working outside the decidable fragment \ref{prop-decidable-epr}, there is no guarantee that such a finite model exists. However, this was never an issue in our experiments. 

We used instead \mathsat \cite{msat} to extract interpolants from ground unrollings, as explained in Section \ref{sub-concretizing-updria}. 
For the second algorithm, denoted as \Lambdatool, we used \icria\cite{msatic3} as a quantifier-free model checker. 
Finally, the sub-procedure of invariant checking with SMT solving explained in \ref{subs-smt-solving}, has been implemented on top of \mathsat. 
Upon termination, both algorithms return either a counterexample trace in a ground instance, or an inductive invariant which can also be checked externally with automatic provers such as Z3, Vampire or CVC5 \cite{z3, Vampire, cvc5}. 

\subsection{Benchmarks}
We collected several benchmarks from different sources, that we divided in three main families. 

The first family of benchmarks consists of 238 array-based systems in \mcmt format, modelling parameterized protocols, timed systems, programs manipulating arrays, and more, taken from \cite{fm18, VERIFY-2010:MCMT_in_Land_of, DBLP:journals/tdsc/BruschiPGLP22} and from the standard \mcmt distribution. 
Often, in those systems, the index theory is pure equality or the theory of a linear order; instead, the theory of element is usually $\mathcal{LIA}$, $\mathcal{LRA}$, or the theory of an enumerated datatype. 
On this set of benchmarks, we could run all the versions of the algorithms presented in this paper, as well as \mcmt. 
In principle, such benchmarks are supported also by \cubicle; 
in practice, however, \cubicle and \mcmt use two (quite) different input languages, so we could only run the former on the 42 instances of this family for which the \cubicle translation is available.

The second family of benchmarks considered are 52 array-based systems in VMT \cite{vmtlib} and \mypyvy format, taken from \cite{ic3po, pdh}. 
This family does not use theories for array elements (i.e. $\mathcal{T}_E$ is simply the theory of Booleans), 
but is more liberal than the previous one in the shape of the formulae used in the transitions.  
Therefore, we could not run \mcmt (nor \cubicle) on this family, 
nor we could use the parameter abstraction technique of Section \ref{subs-param-abs} on them. This family of systems can be used as inputs to the model checkers \icrpo and \mypyvy.

Finally, the third family 
consists of a set of 25 array-based transition systems modeling simple train verification problems \cite{isola20}, on which only the two algorithms presented in this paper can be applied, due to the presence of various SMT theories and non-restricted transition formulae.

\subsection{Comparison of \UPDRIA and \Lambdatool}
\label{subsec-comparison-algos}
We compare here the two algorithms described in this paper, \UPDRIA and \Lambdatool. 
In addition to the basic versions described in Sections~\ref{sec-updria} and \ref{sec-lambda}, for \UPDRIA we also implemented a variant of the algorithm (denoted with option \textbf{--size-$n$}) which combines some of the ideas of Section \ref{sec-lambda}. 
With this option, the algorithm starts by considering a ground instance of size $n$, and extracts a set of lemmas from it (as in \ref{subs-gener}) which can help the algorithm to converge faster or can be discarded if falsified later. For the majority of the benchmarks considered, this strategy (with $n=3$) improved the performances of \UPDRIA.
For \Lambdatool, we consider various options; with the flag \textbf{--no-symm}, we do not use the results of Section \ref{subs-symm} to help the model checking of ground instances. With the flag \textbf{--no-invgen}, we do not apply the invariant minimization technique explained in \ref{subs-minimizing}. 
Moreover, we distinguish three options for implementing the invariant checking sub-procedure of Algorithm \ref{alg-gener}: the option \textbf{--param} consists of the implementation of the parameter abstraction technique of Section \ref{subs-param-abs}; the option \textbf{--ind} applies the procedure of Section \ref{subs-smt-solving}. Finally, the option \textbf{--z3} simply calls \zr on inductive queries.  

We tested all such options on the first family of benchmarks, where it was possible to compare them all. 
A summary of the results is reported in Table \ref{tab:results1-summary}, where we show the number of solved benchmarks and the total time taken by the tool; 
we also report plots comparing different \Lambdatool options in Figure~\ref{fig:cactus-lambda}, 
and different \UPDRIA options in Figure~\ref{fig:cactus-updria}. 
Those plots show, for each point in time, the number of solved instances up to that time.
We consider as solved instances the ones where the algorithms terminate with \Safe or \Unsafe. 
\begin{table}[ht]
  \centering
  \caption{Summary of experimental results on MCMT benchmarks.
  \label{tab:results1-summary}}
\begin{scriptsize}
  \begin{tabular}{l|r|r}
            &  \textbf{Tot solved}   & \textbf{Tot time} \\
  
\hline
\Lambdatool--ind & 216 & 2773s\\
\Lambdatool--nosymm--ind & 216 & 8740s \\
\Lambdatool--z3 & 214 & 11615s\\
\Lambdatool--nosymm--z3 & 214 & 14615s \\
\UPDRIA--size-3 & 208 & 24038s \\
\Lambdatool--param & 206 & 1055s \\
\Lambdatool--noinvgen--ind & 204  & 4073s \\
\Lambdatool--nosymm--param & 203 & 17313s \\
\Lambdatool--noinvgen--param & 202 & 1914s \\
\Lambdatool--noinvgen--nosymm--param & 202 & 10158s \\
\Lambdatool--noinvgen--z3 & 191 & 8897s \\
\Lambdatool--noinvgen--nosymm--ind & 180 & 7986s \\
\UPDRIA--size-2 & 197 & 17488s\\
\UPDRIA--size-4 & 197 & 20943s \\
\Lambdatool--noinvgen--nosymm--z3 & 166 & 3025s \\
\UPDRIA & 162 & 27649s \\

  \end{tabular}
\end{scriptsize}
\end{table}

\begin{figure}[h!]
    \centering
    \includegraphics[width=0.75\textwidth]{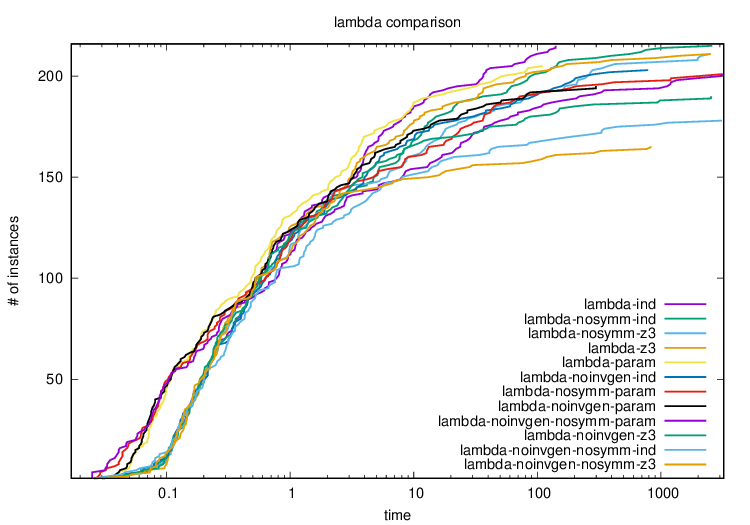}
    \caption{Plot comparing different options of \Lambdatool}
    \label{fig:cactus-lambda}
\end{figure}

\begin{figure}[h!]
    \centering
    \includegraphics[width=0.75\textwidth]{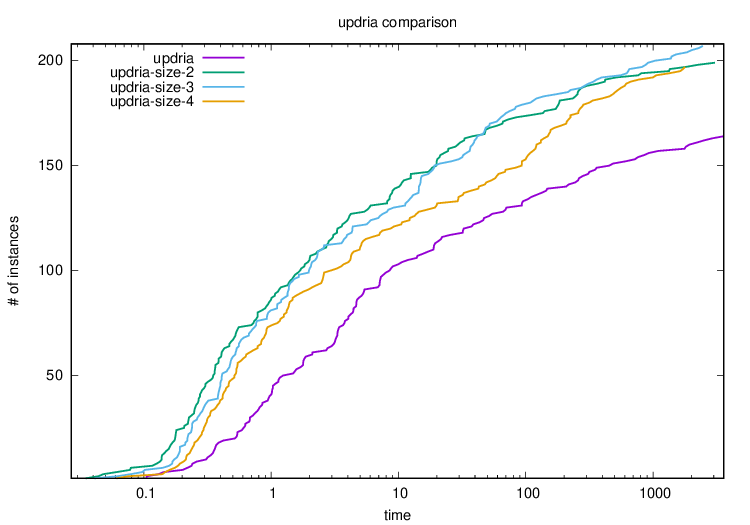}
    \caption{Plot comparing different options of  \UPDRIA}
    \label{fig:cactus-updria}
\end{figure}

In the case of \UPDRIA, most of the time used by the algorithm is consumed by \zr for proving the (un)satisfiability 
of quantified queries. 
In this case, when the procedure diverges, it usually does so by extending the frame sequence without converging to an inductive invariant. 
In general, the algorithm can also diverge with a sequence of refinement failures: as explained in Section~\ref{sub-concretizing-updria}, this can also happen when and all the predicates 
of the interpolants extracted by \mathsat are already in the predicate set $\mathcal{P}$. However, the latter possibility is rarer in our experiments: in the \mcmt benchmarks, this happened only once. Interestingly, we discovered that for that model, no universal inductive invariant exists.

The performance profile of \Lambdatool, instead, is very different from the above. The procedure relies less on quantified reasoning, and most of its time is spent in the model checking of ground instances. 
When the tool diverges, it is usually during this phase: if an invariant is not found, larger and larger (ground) instances are analyzed. 

From the experiments, it was clear that the results about symmetries of Section \ref{subs-symm} reduced drastically the time of model checking of ground instances, thus improving the overall performances. 
The technique of invariant reduction of Section 
\ref{subs-minimizing} was also helpful, in particular, to obtain better lemmas from finite instances (meaning that they 
generalized better to other instances). In this experimental evaluation, the usage of the parameter abstraction technique does not show particular benefits in contrast to the usage of SMT-solving techniques on inductive queries. Although for some instances the computation and model checking of the abstract system was quicker, it is possible to see the prophecy variables of the parameter abstraction as an \textit{a priori} fixed instantiation strategy, thus reducing the generality of the approach. 
Finally, we can also see how our simple instantiation technique was more efficient than a naive usage of \zr: this is because, in the case of non-inductive lemmas, the solver takes a lot of time in building models for counterexamples to induction.
Instead, our approach discards lemmas with an additional ground instance exploration, as explained in Section \ref{subs-smt-solving}.

Finally, we remark on the difference in time consumed by the two algorithms, \UPDRIA and \Lambdatool. As already mentioned, this is due to the fact that most of the \UPDRIA queries are quantified, whereas \Lambdatool tries to avoid that as much as possible. To better illustrate the difference, in Figure \ref{fig:scatter-updria-lambda}, we report a scatter plot comparing the two tools with their respective best options. 

\begin{figure}[h]
    \centering
    \includegraphics[width=0.75\textwidth]{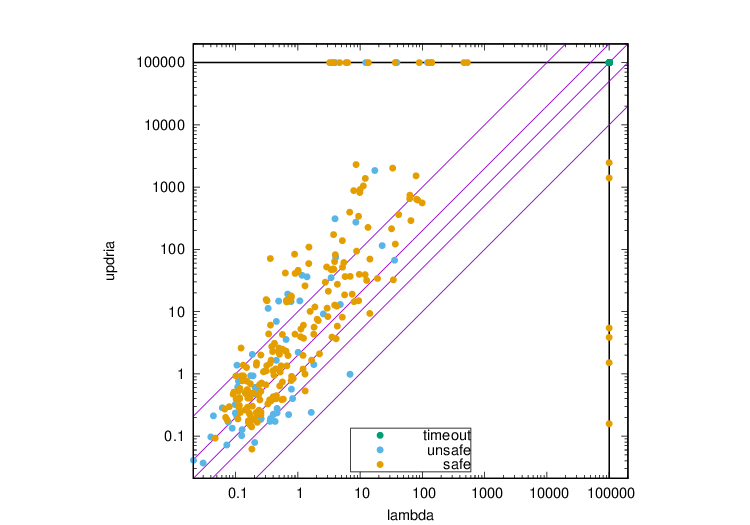}
    \caption{Scatter plot comparing \UPDRIA and \Lambdatool}
    \label{fig:scatter-updria-lambda}
\end{figure}

\subsection{Comparison with other tools}
\label{subsec-comparison-others}
We report in Table \ref{tab:resultsall-summary} a summary of the comparison of the algorithms described in this paper and the other available model checkers. 
The table shows, for each tool/configuration, the number of solved instances for each of the family of benchmarks we considered. In the second column, we show the number of benchmarks solved in the subset of the first family available also for \cubicle. When a tool is not applicable to a family, we use the symbol `\textbf{-}'. For \UPDRIA and \Lambdatool, we used the best options according to the previous section. For \mcmt, we used standard settings: the tool has however different strategies, but we did not investigate the best ones for each benchmark. For \cubicle, we used the \textbf{--brab 2} option. For \mypyvy, we used the implementation and the best options given in the artifact of \cite{QfSep}. For \icrpo, we used the option \textbf{--finv=2} to discharge unbounded checks with Z3. 
The properties of the tools we consider are summarized in Table \ref{tab:tool-summary}.

\begin{table}[ht]
  \centering
  \caption{Summary of the characteristics of the tools considered}
  \label{tab:tool-summary}
\begin{scriptsize}
  \begin{tabular}{l|r|r|r|r}
            & \textbf{$\mathcal{T}_E$ only Boolean} & \textbf{$\mathcal{T}_E$ generic} & \textbf{$\tau(X, X')$ generic} & $\forall \exists$ invariants \\
\hline
\UPDRIA & \checkmark &  \checkmark &  \checkmark & -\\
\Lambdatool &  \checkmark &  \checkmark &  \checkmark & -\\
\mcmt &  \checkmark &  \checkmark & - & -\\
\cubicle &  \checkmark &  \checkmark & - & -\\
\icrpo &  \checkmark & - &  \checkmark  &  \checkmark\\
\mypyvy &  \checkmark & - &  \checkmark &   \checkmark\\
  \end{tabular}
\end{scriptsize}
\end{table}

\begin{table}[ht]
  \centering
  \caption{Summary of experiments on all benchmarks.}
  \label{tab:resultsall-summary}
\begin{scriptsize}
  \begin{tabular}{l|r|r|r|r|r}
            & \textbf{\mcmt (238 tot)} &
            \textbf{\cubicle (42 tot)} &
            \textbf{VMT (52 tot)} & \textbf{Trains (25 tot)} \\
\hline
\UPDRIA & 208 & 30 & 29 & 24 \\
\Lambdatool & 214 & 35 & 29 & 25 \\
\mcmt & 172 & 24  & - & - & \\
\cubicle & - &31 & - & -  \\
\icrpo & - & - & 36 & -  \\
\mypyvy & - & - & 27 & -  \\
\hline
\textbf{VBS} & 220 & 36 & 37 & 25 \\
  \end{tabular}
\end{scriptsize}
\end{table}
We also show, in Figures \ref{fig:mcmt-comp}, \ref{fig:cubicle-comp}, \ref{fig:vmt-compt}, plots comparing \Lambdatool and \UPDRIA with the other tools on single benchmarks families. We can see from the table and the plots the generality of our approach, and the fact that both \UPDRIA and \Lambdatool compare well with the state of the art.

\begin{figure}[h!]
    \centering
    \includegraphics[width=0.5\textwidth]{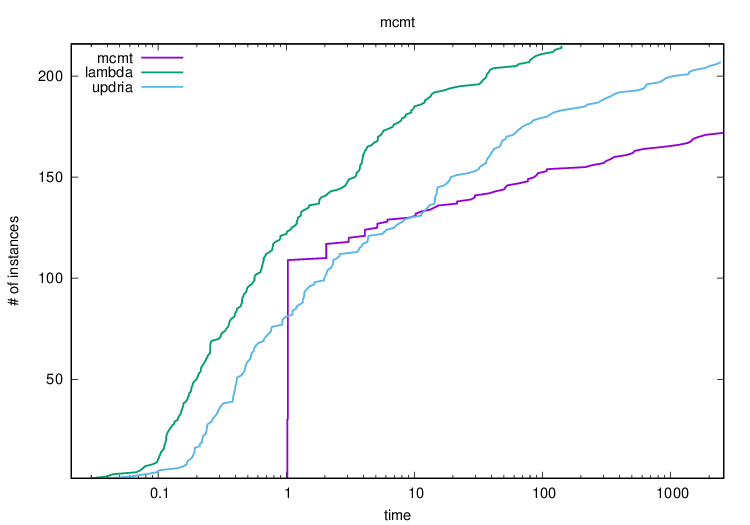}
    \caption{Plot for Mcmt family}
    \label{fig:mcmt-comp}
\end{figure}

\begin{figure}[h!]
\centering
    \centering
    \includegraphics[width=0.5\textwidth]{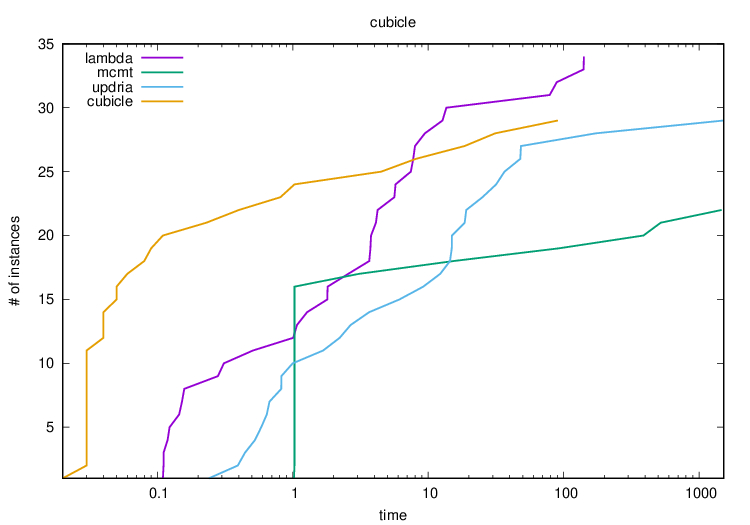}
    \caption{Plot for Cubicle family}
    \label{fig:cubicle-comp}
\end{figure}

\begin{figure}[h!]
    \centering
        \includegraphics[width=0.5\textwidth]{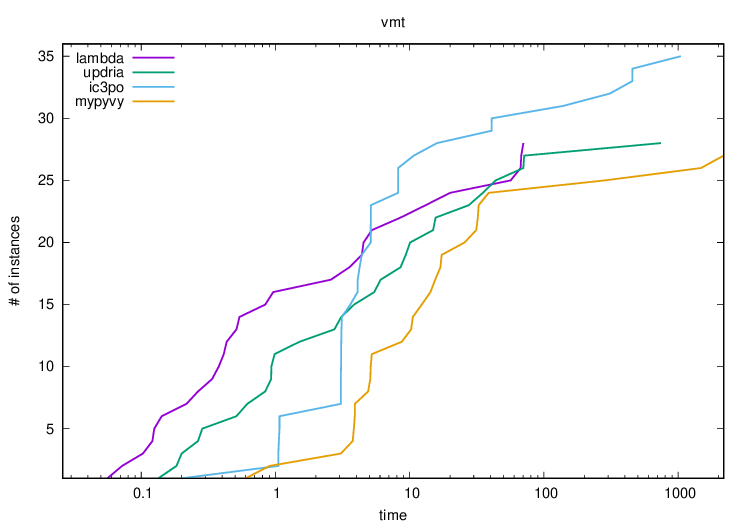}
 \caption{Plot for VMT family}
    \label{fig:vmt-compt}
\end{figure}

In particular, \Lambdatool is outperformed only by \icrpo in the second family of benchmarks, solving 7 more instances. In these cases, \icrpo finds an inductive invariant that contains an existential quantifier, while \Lambdatool (and \UPDRIA) only search for universally quantified inductive invariants. In particular, \UPDRIA diverges on those instances by discovering always the same predicates from a finite instance, as explained in Section~\ref{sub-concretizing-updria}.


\section{Conclusions and Future Work}
\label{sec-conclusions}
In this paper, we studied the problem of synthesizing universal inductive invariants for array-based transition systems, a large class of symbolic transition systems that can be used to model parameterized systems, array programs, and more.
We proposed two algorithms: the first one, called \UPDRIA, combines UPDR (a version of the IC3 algorithm capable of inferring universal invariants) with a novel form of implicit predicate abstraction. 
Once the predicates are fixed, the resulting algorithm returns (upon termination) either a counterexample to the property or a universal invariant over the set of predicates, if exists. However, such an algorithm heavily relies on quantified SMT-solving, which is in practice an expensive subroutine. Moreover, it can diverge with an infinite series of wrong refinements. The second algorithm we propose, called \Lambdatool, is instead more lazy, and relies on the idea that a quantified inductive invariant can be found by generalizing proofs of small instances of the system. This idea requires less quantified reasoning and has proven to be very effective. We have performed an extensive experimental evaluation of the algorithms, that shows that both can be applied to a large set of benchmarks and that they are competitive with other tools. 
As a future work, we plan to extend both algorithms to synthesize invariants with quantifier alternation. Moreover, we intend to investigate the applicability of variants of these algorithms to handle liveness properties. 

\begin{acks}
  This work has been partly supported by the project “AI@TN” funded by the
Autonomous Province of Trento and by the PNRR project FAIR -  Future AI
Research (PE00000013),  under the NRRP MUR program funded by the
NextGenerationEU. 
\end{acks}
\newpage
\bibliographystyle{ACM-Reference-Format}
\bibliography{references}
\newpage
\appendix
\section{Appendix}
In the following, we will refer for the sake of simplicity to a system with a single array state variable $S = (x, \iota(x), \tau(x, x'))$. All the results here can be extended without loss of generality to any finite number of array variables $X = \{x_1, \dots, x_n\}$.
\subsection{Proofs of Section \ref{sec-updria}}
\subsubsection{Proofs of subsection \ref{subs-index-predicate-abs}}
In this section, we start by proving Proposition \ref{prop-index-abs-1}. We consider fixed a set of index predicates $\mathcal{P}(I)$ = \{$p_1(I, x),$ \dots, $p_n(I, x) \}$. Given a system  $S = (x, \iota(x), \tau(x, x,'))$, we denote with $\hat{S}_\mathcal{P} = (X_\mathcal{P}, \hat{\iota}_\mathcal{P}(X_\mathcal{P}), \hat{\tau}_\mathcal{P}(X_\mathcal{P}, X'_\mathcal{P}))$ its indexed predicate abstraction, as defined in \ref{subs-index-predicate-abs}. We also recall the following definitions:
\begin{defn}
Given $\mathcal{P}(I)$ a set of index predicates, we define:
\begin{itemize}
    \item $X_{\mathcal{P}(I)}$ is a set of fresh $\mathcal{T}_I$ predicates, one for each element $ p(I, x) \in \mathcal{P}(I)$, and with ariety the number of free variables in $p$:
    \[  X_{\mathcal{P}(I)} \defeq \{ x_{p(I, x)}(I)\mid p(I, x) \in \mathcal{P}(I) \};\]
    \item the formula \[ H_\mathcal{P}(X_{\mathcal{P}}, x) \defeq \forall I. \big(\bigwedge_{p(I, x) \in \mathcal{P}} x_{p(I, x)}(I) \leftrightarrow p(I, x)\big)  \]
\item the formula
\[ EQ_\mathcal{P}(x, x') \defeq \forall I. \big(\bigwedge_{p(I, x) \in \mathcal{P}} p(I, x) \leftrightarrow p(I, x'))  \]
\end{itemize}
\end{defn}

We start by defining the simulation relation between the abstract and the concrete system. Note that $S$ is defined over the theory \tarray{}, whereas $\hat{S}$ is defined only over the theory $\mathcal{T}_I$. Therefore, a state of $S$ is given by a model $\mathcal{M}$ of \tarray, and a valuation $s$ of the array symbol $x$ as a function. Instead, a state of  $\hat{S}_\mathcal{P}$ is given by an index model $\mathcal{M'}$ and a valuation $\hat{s}$ of the index predicates $X_\mathcal{P}$ as a subset of (a Cartesian product of) $\mathcal{M'}$.
Given a model $\mathcal{M}$ of \tarray, we denote with $\mathcal{M}_{|I}$ its restriction on the index sort. 

\begin{defn}
Let $(\mathcal{M}, s)$ be a state of $ S$, and $(\mathcal{M}', \hat{s})$ a state of $ \hat{S}$. The two states are in the relation $\mathcal{S}_\mathcal{P}$ if and only if $\mathcal{M}_{|I} = \mathcal{M}'$ and $
s, \hat{s} \models H_\mathcal{P}(X_{\mathcal{P}}, x)$.
\end{defn}
We also define the restriction of a state to its abstract counterpart:
\begin{defn}
\label{defn-restriction-index}
    Let $s, \mathcal{M}$ be a state of $S$. We define a state of 
    $\hat{S}_\mathcal{P}$ by considering $\mathcal{M}_{|I}$, and
    $\hat{s}_\mathcal{P}$ a valuation of the predicates $X_P$ into $\mathcal{M}_{|I}$ given by:
    \[ \hat{s}_\mathcal{P}(x_{p(I, x)}) = \{ M \in \mathcal{M}_{|I} \times ... \times \mathcal{M}_{|I}  \ | \ s \models p(M, x) \}.  \]
\end{defn}
    The following proposition follows directly from this definition:
\begin{prop}
\label{prop-restriction-index}
    Let $\mathcal{M}$ a model of \tarray, and $s$ a valuation into $\mathcal{M}$. Let $\hat{s}_\mathcal{P}$ as in Definition \ref{defn-restriction-index}. Then, we have that $s, \hat{s}_\mathcal{P} \in \mathcal{S}.$
\end{prop}

It is now easy to prove that the index predicate abstraction simulates the original system:
\begin{lem}
\label{lem-index-abs-1}
Given $(\mathcal{M}, s)$ a state of $S$ such that $\mathcal{M}, s \models \iota(x)$, then there exists a state $(\mathcal{M}'$, $\hat{s}$) of $\hat{S}_\mathcal{P}$ such that $\mathcal{M}', \hat{s} \models \hat{\iota}_\mathcal{P}(X_\mathcal{P})$, and $ s, \hat{s} \in \mathcal{S}$.
\end{lem}
\begin{proof}
    Recall that $\hat{\iota}_{\mathcal{P}}(X_{\mathcal{P}}) \defeq \exists x. (\iota(x) \wedge  H_\mathcal{P}
(X_{\mathcal{P}}, x) )$.     Let $\mathcal{M}' = \mathcal{M}_{|I}$, and $\hat{s} = \hat{s}_\mathcal{P}$ as defined in the previous definition. From \ref{prop-restriction-index}, we have that $s, \hat{s} \in \mathcal{S}$, and the lemma follows.
\end{proof}
Similarly, we can prove that:
\begin{lem}
\label{lem-index-abs-2}
    Let $\mathcal{M}$ a model of \tarray, and let $s, \hat{s}$ a couple of states such that $ s, \hat{s} \in \mathcal{S}$. Then, for every $s'$ such that $\mathcal{M}, s, s' \models \tau(x, x')$, there exists an $\hat{s}'$ such that $\mathcal{M}_{|I}, \hat{s}, \hat{s}' \models \hat{\tau}_\mathcal{P}(X_\mathcal{P}, X'_\mathcal{P}).$
\end{lem}
We already mentioned that the index predicates abstraction preserves the validity of all the formulas whose atoms are contained in $\mathcal{P}$. Therefore, we have the following (stated as Proposition \ref{prop-index-abs-1} in Section \ref{sec-updria}): 
\begin{prop}
\label{prop-index-abs-app}
   Let \system an array-based transition system, and $\phi(X)$ a formula. Let $\mathcal{P}$ a set of index predicates which contains all the atoms occurring in $\phi$. If $\hat{S}_\mathcal{P} \models \hat{\phi}$ then $S \models \phi$.
\end{prop}

\begin{proof}
    By Lemmas \ref{lem-index-abs-1}, \ref{lem-index-abs-2}, we have that the relation $\mathcal{S}_\mathcal{P}$ is a simulation relation, and therefore it preserves reachability (\ref{subs-simul}). Suppose that  $S \models \neg \phi$: then, there exists a finite path to a state $s$ such that $s \models \neg \phi$. By using definition \ref{defn-restriction-index}, it follows that there exists a reachable state $\hat{s}$ of $\hat{S}$ such that $s, \hat{s} \in \mathcal{S}_\mathcal{P}$. Since $\mathcal{P}$ contains all the atoms of $\phi$, we have that $\hat{s} \models \neg \hat{\phi}$, a contradiction.
\end{proof}
Moreover, can now establish the following fact. Recall that
\begin{align*}
     AbsRelInd(F, \tau, \psi, \mathcal{P}) =  F(X_{\mathcal{P}}) \wedge \psi(X_{\mathcal{P}})\wedge H_\mathcal{P}(X_{\mathcal{P}}, X)
      \\
      {} \wedge EQ_\mathcal{P}(X, \bar{X}) \wedge \tau(\bar{X}, \bar{X}')\wedge EQ_\mathcal{P}(\bar{X}', X') \wedge \neg \psi(X_{\mathcal{P}}') \wedge H_\mathcal{P}(X_{\mathcal{P}}', X')  
    \end{align*}

We have: 
\begin{prop}
\label{prop-abs-relative}
Given a system \system and its abstraction $\hat{S}_\mathcal{P} = $ $(X_\mathcal{P}, $ $\hat{\iota}_\mathcal{P}(X_\mathcal{P}),$ $ \hat{\tau}_\mathcal{P}(X_\mathcal{P}, X'_\mathcal{P}))$, given any formulae $F, \psi$, then the formulae
    $AbsRelInd(F, \tau, \psi, \mathcal{P})$ and $F(X_{\mathcal{P}}) \wedge \hat{\tau}(X_{\mathcal{P}}, X'_{\mathcal{P}})$ $ \wedge \psi(X_{\mathcal{P}}) \wedge $ $\neg \psi(X'_{\mathcal{P}})$ are equisatisfiable. Moreover, if $s, s' \models AbsRelInd(F, \tau, \psi, \mathcal{P})$, then $\hat{s}_\mathcal{P}, \hat{s'}_\mathcal{P} \models F(X_{\mathcal{P}}) \wedge \hat{\tau}(X_{\mathcal{P}}, X'_{\mathcal{P}}) \wedge \psi(X_{\mathcal{P}}) \wedge \neg \psi(X'_{\mathcal{P}})$. 
\end{prop}
\begin{proof}(sketch)
The proof is similar to Theorem 1 in \cite{msatic3}, by using Definition \ref{defn-restriction-index} instead of the projection on Boolean values.
\end{proof}

Proposition \ref{updr-abstract-cex} is an immediate consequence of the latter, once noticing that we compute diagrams from the restriction of models over the signature $\Sigma_I \cup X_\mathcal{P}$: every step of the algorithm is equivalent to performing UPDR \cite{updr} on $\hat{S}_\mathcal{P}$.

\subsubsection{Proofs of subsection \ref{subs-updria-alg}}
 We give here the proof of the Proposition \ref{prop-updria-soundness}. First, we have
\begin{prop}
    The frames $F_0(X_\mathcal{P}), \dots, F_n(X_\mathcal{P})$ are an approximate reachability sequence for $\hat{S}_\mathcal{P}$.
\end{prop}
\begin{proof}(sketch)
    The proof is the same as the one in \cite{msatic3}, Lemma 1.
\end{proof}

\begin{prop}
    If algorithm \ref{updralg} terminates with \Safe, then $S \models \phi$. If algorithm \ref{updralg} terminates with \Unsafe, then $S \not \models \phi$
\end{prop}
\begin{proof}
    The algorithm terminates with \Safe when there exists an approximate reachability sequence $F_0(X_\mathcal{P}), \dots, F_n(X_\mathcal{P})$ such that, for some $i$, $F_{i+1} \models F_i$. From Proposition \ref{prop-approxi-sequence}, we have that $\Tilde{S}_\mathcal{P} \models \hat{\phi}$. From Proposition \ref{prop-index-abs-app}, this implies that $S \models \phi$. If the algorithm terminates with \Unsafe, then the unrolling \ref{concrete-unrolling} is satisfiable, and we have a counterexample.
\end{proof}

\subsection{Proofs of Section \ref{subs-param-abs}}
\label{sec-appendix}
We report here the technical results for the proof of Proposition~\ref{prop-general-simulation}. We consider an array-based transition system $S = (x, \iota(x), \tau(x, x'))$, a candidate property $\phi(x)$, and its parameter abstraction $ \tilde{S} = (\{\tilde{x}, P, E\}, \tilde{\iota}(P, \tilde{x}), \tilde{\tau}(P, P', E, \tilde{x}, \tilde{x}'))$, as defined in \ref{subs-param-abs}, where the $\tilde{x}$ are a renaming of the $x$. Note that a state $\tilde{s}$ of $\tilde{S}$ consists of both an assignment of the array variable $x$ and of the index variables $P \cup E$. With a small abuse of notation, we do not distinguish the two cases. The simulation of the abstraction is given by the formula:
\begin{defn}
    Let $S$ an array based transition system and $\tilde{S}$ its parameter abstraction. Let $P = \{ p_1, \dots, p_n \}$ be the set of prophecy variables. We define
    \[ \tilde{H}(x, \tilde{x}) \defeq \bigwedge_{i=1, \dots, n} \Tilde{x}[p_i] = x[p_i] \]
    Let $\mathcal{M}$ be a model of \tarray. If $s$ is a state of $S$, and $\tilde{s}$ of $\Tilde{S}$, we define a relation $\mathcal{S}$ among states in this way: 
    \[ s, \tilde{s} \in \mathcal{S} \Leftrightarrow \mathcal{M}, s, \Tilde{s} \models \Tilde{H}(x, \tilde{x}). \]
    
\end{defn}
First, we prove the following Proposition.
\begin{prop}
\label{equiv-prop} Let $\mathcal{M}$ a model of \tarray.
Let $\tilde{s}$ be a state of $\tilde{S}$. Let $\mu$ be an interpretation of $P$ such that $\mu(P) = \tilde{s}(P)$. Let $\phi(P, x)$ be a quantifier free formula which contains only prophecies as free index variables.
  Then, for any state $s$ of $S$ such that $\mathcal{S}(s,\tilde{s})$,
  $$
  \tilde{s} \models \phi(P, x)  \Leftrightarrow s, \mu \models \phi(P, x)
  $$
\end{prop}
\begin{proof}	
Note that a model for a function is uniquely determined by the values on its domain. So, if $\mathcal{M}$ is a model for the total functions from $\mathcal{M}_I$ to $\mathcal{M}_E$, then 
\[\mu, s \models \phi(P, x) \]
where $s$ is a valuation to into $\mathcal{M}$, is equivalent to 
\begin{equation}
\label{cond-1}
\mu, s' \models \phi(P, x),
\end{equation}
where $s'$ is a valuation to $\mathcal{N}$, which obtained from $\mathcal{M}$ by restricting all the interpretation of the index variables to the substructure of $\mathcal{M}_I$ generated by the elements in $\mu(P)$. Similarly, for any model $\mathcal{M}'$ and valuation $\tilde{s}$ into it, 
\[\tilde{s} \models \phi(P, x) \]
is equivalent to 
\begin{equation}
\label{cond-2}
\tilde{s}' \models \phi(P, x),
\end{equation}
where $\mathcal{N'}$ defined similarly as above. Since $\mu(P) = \tilde{s}(P)$, we have $\mathcal{N} = \mathcal{N'}$. Moreover, from the definition of $\mathcal{S}$, $s'$ and $\tilde{s}'$ assign $x$ to the same function, so (\ref{cond-1}) and (\ref{cond-2}) are equivalent. 
\end{proof}

\begin{lem}
\label{lem-simulation-initial} Let $\mathcal{M}$ be a model of \tarray.
If $ s \models \iota(x)$, then there exists some $\tilde{s}$ such that
$\mathcal{S}(s,\tilde{s})$ and $\tilde{s}\models \tilde{\iota}(P, x)$.
\end{lem}

\begin{proof}
Let $s$ be an assignment into a model $\mathcal{M}$, with index domain $\mathcal{M}_I$, and let $m$ be the length of $I$. 
Then, \[\tilde{\iota}(P, x) = \bigwedge_{p_{i_1}, \dots, p_{i_m} \subseteq P^m} {\phi}(p_{i_1}, \dots, p_{i_m}, x[p_{i_1}, \dots, p_{i_m}] ). \] Let $\tilde{\mu}$ an (injective) assignment of the prophecy variables $P$ into $\mathcal{M}_I$. Let $\tilde{s}$ defined as the restriction of $s$ over $\tilde{\mu}(P)$ and such that $\tilde{s}(P) = \tilde{\mu}(P)$. By definition, $(s, \tilde{s}) \in \mathcal{S}$. Then, by Proposition \ref{equiv-prop}, we have
$$\tilde{s} \models \tilde{\iota}(P, x) \Leftrightarrow s, \tilde{\mu} \models \tilde{\iota}(P, x).$$

Since the formula $\iota(x) \rightarrow \tilde{\iota}(P, x)$ is valid, and $s \models \iota(a)$ by hypothesis, the claim follows. 
\end{proof}
\begin{lem}
\label{lem-simulation-transition} Let $\mathcal{M}$ be a model of \tarray. If $s, s'\models\tau(x,x')$, then for every $\tilde{s}$ such that $(s,\tilde{s}) \in \mathcal{S}$, either:
\begin{itemize}
\item there exists a rule $\tilde{\tau}$ and some $\tilde{s}'$, 
such that $\tilde{s},  \tilde{s}' \models \tilde{\tau}$ and $(s', \tilde{s}') \in \mathcal{S}$; or
\item there exist a rule $\tilde{\tau}$ and some $\tilde{s}'$, 
$\tilde{s}''$, such that $\tilde{s}, \tilde{s}' \models \tilde{\tau}_S$, 
$\tilde{s}', \tilde{s}'' \models \tilde{\tau}$, 
and $(s', \tilde{s}'') \in \mathcal{S}$.
\end{itemize}
\end{lem}

\begin{proof}
 We first consider the simpler case of one prophecy variable $p$ and one environmental variable $e$. 
  By hypothesis,
  $$
      s, s'\models\exists i \forall J. \psi(i, J, a, a[J],  a', a'[J]).
$$
  Hence, there exists an interpretation $\mu$ of $i$ in an element of $\mathcal{M}_I$ such that
  \begin{equation}
  \label{1}  
  s, s', \mu \models \forall J. \psi(i, J, a, x[J], x', x'[J]).
  \end{equation}
Let's also fix a state $\tilde{s}$ of $\tilde{S}$, such that $\mathcal{S}(s, \tilde{s})$. There are now three cases:
\begin{itemize}
\item Suppose $\mu(i) = \tilde{s}(p)$. Then, the transition of $\tilde{S}$ labeled by the substitution  $i \mapsto p$ is:
$$\tilde{\tau}_{\sigma : i \mapsto p} = \bigwedge_{j \in {p, x}} \psi(p, J, x, x[J], x', x'[J]).$$
Let $\tilde{s}'$ defined as $\tilde{s}'(p) \defeq \mu(i)$ and $\tilde{s}'(x)[\tilde{s}'(p)] \defeq s'(x)[\mu(i)]$. Note that $\mathcal{S}(s', \tilde{s}')$ by definition. Since $(\ref{1})$ is universal and  $\mu(i) = \tilde{s}(p)$, with an argument similar to lem~\ref{lem-simulation-initial}, we have that $\tilde{s}, \tilde{s}' \models \tilde{\tau}_{\sigma : i \mapsto p}$. 

\item Suppose $\mu(i) \neq \tilde{s}(p)$ but $\mu(i) = \tilde{s}(e)$ and $\tilde{s}(x)[\tilde{s}(e)] = s(x)[\mu(i)]$. Then, consider the transition labeled by the substitution $i \mapsto e$. Similarly to the first case, we can define $\tilde{s}'$ to be the restriction of $s'$ over $p$ and $e$, and we have that $\tilde{s}, \tilde{s}'\models \tilde{\tau}_{\sigma : i \mapsto e}$. Moreover, $\mathcal{S}(s', \tilde{s}')$ by definition.
 
\item If instead $\mu(i) \neq \tilde{s}(x)$ or $\tilde{s}(x)[\tilde{s}(e)] \neq s(x)[\mu(i)]$, we can reduce to the previous case with a stuttering transition. Let $\tilde{s}'$ defined as $\tilde{s}$ on $p$, but $\tilde{s}'(e) \defeq \mu(i)$ and $\tilde{s}'(x)[\tilde{s}(e)] \defeq s(x)[\mu(i)]$. Note that we also have $(s, \tilde{s}') \in \mathcal{S}$.
Then $\tilde{s}, \tilde{s}' \models \tilde{\tau}_S$, and we have reduced to the previous case.  So, there exists an $\tilde{s}''$ such that $\tilde{s}', \tilde{s}''\models \tilde{\tau}_{\sigma : i \mapsto e}$ and $\mathcal{S}(s', \tilde{s}'')$. 
\end{itemize}

In general, suppose $P = (p_1, \dots, p_n)$. By hypothesis,
$$
      s, s' \models\exists I \forall J. \psi(i, J, x, x[J],  x', x'[J]).
$$
Since the length of $E$ is the maximum length of the existentially quantified index variables in the rules of $S$, we can assume without loss of generality that $I = (i_1, \dots, i_m)$ and $E = (e_1, \dots, e_m)$. By hypothesis there exists an interpretation $\mu$ of $I$ such that 
$$
s, s', \mu \models \forall J. \psi(I, J, x, x[J], x', x'[J]).
$$
Let's also fix a state $\tilde{s}$ of $\tilde{S}$, such that $\mathcal{S}(s, \tilde{s})$. There are again three cases; we omit the details since they are a generalization of the previous ones.  
\begin{itemize}
\item if $\mu(I) \subseteq \tilde{s}(P)$, then there exist $P_{J} = ( p_{j_1}, \dots, p_{j_m})$ such that $\mu(I) = \tilde{s}(P_{J})$. We can define $\tilde{s}'$ to be the restriction of $s$ over, $P$ and we have again $\tilde{s}, \tilde{s}' \models \tilde{\tau}_{\sigma : I \mapsto P_{J}}$. 
\item Suppose now there exists a $0 \leq h < m$ such that $\mu(i_1, \dots, i_h) = \tilde{s}(p_{j_1}, \dots, p_{j_h})$, and moreover $\mu(i_{h+1}, \dots, i_{m}) = \tilde{s}(e_1, \dots, e_{m-h})$, and also $\tilde{s}(a)[\tilde{s}(e_1, \dots, e_{m-h})] = s(a)[\mu(i_{h+1}, \dots, i_{m})]$. Then, if we define $\tilde{s}'$ to be the restriction of $s'$ over $P \cup E$, we have that $\tilde{s}, \tilde{s}' \models \tilde{\tau}_{\sigma}$ where $\sigma : I \mapsto \{p_{j_1}, \dots, p_{j_h}, e_1, \dots, e_{m-h}\}$.
\item If instead $\mu(i_{h+1}, \dots, i_{m}) \neq \tilde{s}(e_1, \dots, e_{m-h})$ or $\tilde{s}(a)[\tilde{s}(e_1, \dots, e_{m-h})] \neq s(a)[\mu(i_{h+1}, \dots, i_{m-h})]$, we can reduce to the previous case with a stuttering transition. Let $\tilde{s}'$ defined as $\tilde{s}$ on $P$, but $\tilde{s}'(e_1, \dots, e_{m-h}) \defeq \mu(i_{h+1}, \dots, i_{m})$ and $\tilde{s}'(a)[\tilde{s}(e_1, \dots, e_{m-h})] \defeq s(a)[\mu(i_{h+1}, \dots, i_{m-h})]$. Note that $\mathcal{S}(s, \tilde{s}')$ by definition. Moreover, $\tilde{s}, \tilde{s}' \models
\tau_S$. We have now reduced to the previous case, and the claim follows.
\end{itemize}

\end{proof}

\begin{thm}
The relation $\mathcal{S}$ is a stuttering simulation between $S$ and $\tilde{S}$. \end{thm}
\begin{proof}
Follows directly from
Lemmas~\ref{lem-simulation-initial}
and~\ref{lem-simulation-transition}.  
\end{proof}

\begin{thm}
\label{thm-stuttering-simulation}
Let $S$ be an array-based transition system, $\tilde{S}$ its parameter abstraction. Let $\forall I. \Phi(I, a)$ a candidate invariant, and $P$ a set of frozen variables with same length as $I$.
If $\tilde{S} \models \Phi(P, a)$, then $S \models \Phi(P, a)$
\end{thm}
\begin{proof}
The statement follows from the fact that stutter simulations preserve reachability (\ref{subs-simul}), and from Proposition \ref{equiv-prop}.

\end{proof}

\newpage

\end{document}